\documentclass[12pt]{article}
\usepackage{geometry}
\usepackage{amsthm, amsmath, amssymb, bbm, bm, mathrsfs}
\usepackage[round]{natbib}
\usepackage[colorlinks, 
			linkcolor=blue,
			anchorcolor=blue,
			urlcolor=blue,
			citecolor=blue]{hyperref}
\usepackage{graphicx}
\usepackage{mathrsfs}
\usepackage{threeparttable}
\usepackage[title]{appendix}
\usepackage{url}
\usepackage{diagbox}
\usepackage{booktabs}
\usepackage{caption}
\usepackage{subfigure}
\usepackage{float}
\usepackage{listings}
\usepackage{multirow}
\usepackage{rotating}
\usepackage{longtable}
\usepackage{enumitem}
\usepackage{cases}
\usepackage{algorithm}
\usepackage{algorithmic}
\usepackage{ulem}

\usepackage{xr}

\newtheorem{theorem}{Theorem}[section]
\newtheorem{corollary}{Corollary}[section]

\newtheorem{proposition}{Proposition}[section]

\newtheorem{lemma}{Lemma}[section]

\newtheorem{assumption}{Assumption}

\newtheoremstyle{mycase}{5pt}{5pt}{\upshape}{}{\bfseries}{.}{ }{} \theoremstyle{mycase}

\newtheoremstyle{myexample}{5pt}{5pt}{\upshape}{}{\bfseries}{.}{ }{} \theoremstyle{myexample}

\renewcommand{\theequation}{\thesection.\arabic{equation}}
\numberwithin{equation}{section}

\renewcommand{\hat}{\widehat}
\renewcommand{\tilde}{\widetilde}

\newcommand{\Var}{\mathrm{Var}}
\newcommand{\bX}{\mathbf{X}}
\newcommand{\bx}{\mathbf{x}}
\newcommand{\bbeta}{\bm{\beta}}
\newcommand{\btheta}{\bm{\theta}}
\newcommand{\bgamma}{\bm{\gamma}}
\newcommand{\bkappa}{\bm{\kappa}}
\newcommand{\bdelta}{\bm{\delta}}
\newcommand{\bzeta}{\bm{\zeta}}
\newcommand{\bxi}{\bm{\xi}}
\newcommand{\mbR}{\mathbb{R}}

\newcommand{\mbE}{\mathbb{E}}
\newcommand{\mbB}{\mathbb{B}}
\newcommand{\mbC}{\mathbb{C}}
\newcommand{\mcD}{\mathcal{D}}
\newcommand{\mcL}{\mathcal{L}}
\newcommand{\mbA}{\mathbf{A}}
\newcommand{\bu}{\bm{u}}
\newcommand{\bv}{\bm{v}}

\newcommand{\bSigma}{\mathbf{\Sigma}}

\newcommand{\mbX}{\mathbf{X}}
\newcommand{\bindicator}{\mathbb{I}}

\begin{document}
\allowdisplaybreaks[3] 
	
\title{\bf Renewable high-dimensional expected shortfall regression}
\author{Haochen Rao$^1$, Tingzi Weng$^1$, Yifan Jiang$^2$ and Xu Guo$^1$\thanks{Corresponding author: Xu Guo. Email address: xustat12@bnu.edu.cn.}\\~\\
	{\small \it $^{1}$ School of Statistics, Beijing Normal University, Beijing, China} \\
    {\small \it $^{2}$ Department of Statistics, Pennsylvania State University, Pennsylvania, United States
     }
}

\date{}
\maketitle

\vspace{-0.5in}
	
\begin{abstract}

Expected Shortfall (ES) has become a core coherent risk measure in finance and statistics, and high-dimensional ES regression is crucial for characterizing heterogeneous tail risk with massive covariates.  Existing offline methods for high-dimensional ES regression rely on access to full data, which fails under streaming data scenarios with sequential batch arrival and limited storage. To address this issue, this paper proposes a renewable estimation and inference framework for high-dimensional ES regression tailored to streaming data. By optimizing a surrogate loss function determined only by current data and historical information, the proposed procedure updates the estimator of ES regression coefficients without storing full raw data. Based on the online estimator, we design an online debiased estimator and further construct valid Wald-type confidence intervals using consistent variance estimation. Theoretically, we establish non-asymptotic error bounds for the online high-dimensional ES estimator and verify the asymptotic normality of the online debiased estimator. Extensive simulations show that the proposed method achieves estimation accuracy and inference performance comparable to the offline benchmark. Moreover, an application on the car insurance claim dataset demonstrates strong practical value in insurance risk management.

\medskip
		
	
\end{abstract}

\section{Introduction}

In modern financial risk management, accurate tail risk measurement has become a core academic and practical issue. The limitations of traditional Value-at-Risk (VaR) have been exposed in several recent financial crises. Subsequently, the \textbf{Expected Shortfall (ES)} has risen as a coherent risk measure \citep{acerbi2002oncohES} and has become the mainstream market risk supervision standard in the Basel Accord    (\url{https://www.bis.org/bcbs/publ/d457.htm}). The ES measure has been widely used in financial regulation \citep{Embrechts2014Risks,Zaevski2023IRFA}, portfolio optimization \citep{Krokhmal2001JoRisks,DeMiguel2009Opt}, and insurance pricing \citep{Frees2013Book,Brownlees2017SRISK}. Instead of quantifying the loss at a specified level, ES measures the conditional expectation of losses exceeding the quantile threshold. Compared with VaR, ES provides a more comprehensive assessment of tail risk, leading to extensive research on estimation and inference methods in statistical and financial literature.

The basic framework of coherent risk measures laid the theoretical foundation for ES (also referred to as Conditional Value-at-Risk, CVaR) \citep{Artzner1999MF}. \cite{Rockafellar2000OptCVaR,Rockafellar2002JBF} further proposed the optimization and regression framework for CVaR and provided its convexity and tractability. However, estimating ES in the regression setting is challenging due to the lack of elicitability \citep{Gneiting2011JASA}, which means that no loss exists as the objective function such that ES is the minimizer of it. To tackle this issue, \cite{Dimitriadis2019Joint} developed a joint regression framework for the conditional quantile and the ES, based on the statement that the quantile and the ES are jointly elicitable \citep{Fissler2016AOS}. According to this methodology,  \cite{he2023robustES} further developed a two-step ES regression framework and investigated its asymptotic properties under a moderate-dimensional setting where $p=O(n^{\kappa})$ for some $\kappa\in(0,1)$. With the explosion of big data, high-dimensional ES regression has attracted increasing attention. The most groundbreaking work was \cite{zhang2023hdesreg}, which proposed a two-step $\ell_1$-regularized estimation framework for high-dimensional ES regression and established its non-asymptotic error bounds and decorrelated score inference procedure. \cite{gao2025lintestES} further extended this framework to linear hypothesis testing for high-dimensional ES regression with heavy-tailed errors. Despite the advances, these works all rely on full access to historical data, making them inapplicable to streaming data scenarios where data arrives sequentially in batches. 

In the field of high-dimensional statistics for streaming data, \cite{Bottou2008Nips,Duchi2009} first introduced regularization into high-dimensional online learning. For statistical inference with streaming data, \cite{han2021delasso} developed the online debiased lasso estimator, and \cite{luo2021glm} extended this to generalized linear models. \cite{Han2024onlinesgd} further introduced an online debiased stochastic gradient descent method for linear models. For quantile regression, which is closely linked to ES modeling, \cite{Xie2023onlinesmoothquantile} proposed a renewable estimation and inference method for smoothed quantile regression. Moreover, online estimation and inference procedures have been developed for various models, see \cite{wu2021survive,quan2024onepass,jiang2024renewable,rao2025onlinesvm}. For ES regression, the most recent work \cite{Mi2024twosteponlineES} proposed an estimation framework adaptable to online updating. However, they considered only low-dimensional covariates. Gaps still remain at the intersection between high-dimensional ES regression and streaming data analysis. On one hand, existing renewable methods cannot be directly extended to ES regression due to its unique two-step procedure, where the estimation error of the quantile regression in the first step may propagate and accumulate across batches. On the other hand, the orthogonal score of ES regression requires a specially designed estimator to correct extra bias from online updates, requiring more complicated theories than standard linear models. 

To address these challenges, this paper proposes a renewable estimation and inference framework for high-dimensional expected shortfall regression with streaming data. Our main contributions are listed as follows: 

(1). We develop a renewable online estimator for high-dimensional ES regression coefficients. Applying Taylor's expansion, we construct a surrogate loss function relying solely on current batch data and historical summary statistics and free of full historical data. We further establish its non-asymptotic error bounds in terms of $\ell_{\bSigma}$-norm and $\ell_1$-norm, showing that its convergence rate matches the offline estimator with full sample. 

(2). We construct a valid online inference procedure for high-dimensional ES regression. We propose a bias-corrected online debiased estimator for ES coefficients and rigorously prove its asymptotic normality. Furthermore, Wald-type confidence intervals are constructed provided the consistent variance estimators.

(3). We solve the error propagation and accumulation problem in the two-step online estimation. Based on the precise error control conditions for the online quantile regression estimator, we prove that the estimation error will not accumulate with batch updating as long as the initial batch size meets the growth conditions. 

(4). We verify the finite sample performance of the proposed method through extensive simulations and a real data application on car insurance claims. The results show that our method achieves almost the same accuracy as the offline benchmark under storage constraints, with strong practical value in real-time risk management.

The rest of this paper is organized as follows. Section 2 introduces the proposed methodology, including a recap of the offline high-dimensional ES regression framework, the renewable online estimation procedure, and the online debiased inference method. Section 3 presents the main theoretical results for both online estimation and inference procedures. Section 4 conducts numerical studies, including Monte Carlo simulations and a real data application. Section 5 concludes the paper. Detailed proofs of the main theorems, technical lemmas, and additional numerical results are provided in the Appendix.

{\bf Notations.} For a vector $\bm{a}=(a_1, \dots, a_p)^\top \in\mbR^p$, denote $\Vert \bm{a} \Vert_0= \sum_{j=1}^p I(a_j\neq 0)$, $\Vert \bm{a}\Vert_r=(\sum_{j=1}^p|a_j|^{r})^{1/r}$ for $r\geq 1$, and $\Vert \bm{a} \Vert_{\infty}=\max_{j=1,\dots ,p}|a_j|$.  For a real matrix $\mbA=\{a_{ij}\}_{i,j}^p \in \mbR^{p\times p}$, let $\tau_{\max}(\mbA)$ and $\tau_{\min}(\mbA)$ be the maximum and the minimum eigenvalues of $\mbA$, respectively. Define the elementwise max-norm $\Vert \mbA\Vert_{\max}=\max_{i,j}|a_{ij}|$, the elementwise $\ell_1$-norm $\Vert \mbA\Vert_{1}=\sum_{i,j}|a_{ij}|$, the spectral norm $\Vert \mbA\Vert_{2}=\sqrt{\tau_{\max}(\mbA^\top\mbA)}$ and the matrix $L_1$-norm $\Vert \mbA\Vert_{L_1}=\max_j\sum_{i}|a_{ij}|$. Moreover, for a vector $\mathbf{u}\in\mbR^d$ and a symmetric matrix $\mathbf{A}\in\mbR^{d\times d}$, $\Vert\mathbf{u}\Vert_{\mathbf{A}}=\Vert\mathbf{A}^{1/2}\mathbf{u}\Vert_2$.  
For two sequences of non-negative numbers $a_n$ and $b_n$, we denote $a_n=o(b_n)$ if $\lim_{n\to\infty} a_n/b_n=0$, and $a_n=O(b_n)$ or $a_n\lesssim b_n$ if $|a_n|\leq Cb_n$ for some positive constant $C\geq 0$ independent of $n$, $a_n \gtrsim b_n$ is similarly defined, and $a_n\asymp b_n$ if $a_n=O(b_n)$ as well as $b_n=O(a_n)$.

\section{Methodology}
\subsection{Existing offline methods}
We first introduce the basic conceptions. Let $Z$ be a real-valued random variable and $F_Z:\mbR\to[0,1]$ be its cumulative distribution function. At a pre-specified quantile level $\tau\in (0,1)$, define the quantile function $Q_{\tau}(Z) = \inf\{z\in\mbR:F_Z(z)\geq \tau\}$. Furthermore, the lower $\tau$-th ES function is defined as
\begin{align*}
    ES_{\tau}(Z) = \frac{1}{\tau}\int_{0}^{\tau}Q_u(Z) {\rm d}u.
\end{align*}
By definition, the lower $\tau$-th ES of $Z$ summarizes the tail average of $Z$ below the $\tau$-th quantile. The upper $\tau$-th ES is symmetrically defined as the average of the tail over $[\tau,1)$. 

Suppose we observe a random sample $\{ Y_i, \bX_i\}_{i=1}^n$ from $(Y, \bX)\in\mbR\times\mbR^p$. Consider the joint regression of the conditional quantile and expected shortfall of $Y_i$ given $\bX_i$:
\begin{align}
    \label{jointqtesreg}
    Q_{\tau}(Y_i\mid\bX_i) = \bX_i^\top \bbeta^*(\tau), \qquad ES_{\tau}(Y_i\mid\bX_i) = \bX_i^\top \theta^*(\tau), \qquad i=1,\dots,n,
\end{align}
where $\bbeta^*(\tau), \btheta^*(\tau)$ are the true parameters of the quantile regression and the expected shortfall regression, respectively.   Under the high-dimensional regime where $p$ diverges faster than $n$, \cite{zhang2023hdesreg} considered a two-step $\ell_1$-regularized estimation framework to estimate $\btheta^*(\tau)$. Firstly, they estimate $\bbeta^*(\tau)$ through penalized quantile regression:
\begin{align}
    \label{offpenquantest}
    \hat{\bbeta}_{\lambda_q}(\tau) = \arg\min_{\bbeta\in\mbR^p} \frac{1}{n}\sum_{i=1}^n \rho_{\tau}(Y_i-\bX_i^\top\bbeta) + \lambda_q \Vert \bbeta \Vert_1.
\end{align}
Here $\rho_{\tau}(u) = u\{\tau - \mathbf{1}{(u<0)}\}$ is the check loss function. Then, using the estimator obtained from \eqref{offpenquantest}, they estimate $\btheta^*$ by minimizing the penalized orthogonal score as
\begin{align}
    \label{offpenESest}
    \hat{\btheta}_{\lambda_e}(\tau) = \arg\min_{\btheta\in\mbR^p} \frac{1}{2n} \sum_{i=1}^n \Big[Z_i\big\{\hat{\bbeta}_{\lambda_q}(\tau)\big\} -\tau\bX_i^\top\btheta \Big]^2 + \tau\lambda_e \Vert\btheta\Vert_1,
\end{align}
where $Z_i(\bbeta)=(Y_i-\bX_i^\top\bbeta)\bindicator(Y_i\leq \bX_i^\top\bbeta) + \tau\bX_i^\top\bbeta$. For convenience, we omit the hyper-parameter $\tau$ and briefly write $\bbeta^*,\hat{\bbeta}$ and $\btheta^*,\hat{\btheta}$ throughout our analysis. The rationale behind  procedures \eqref{offpenquantest} and \eqref{offpenESest} is the orthogonality of the score function $\psi(\bbeta,\btheta;\bX) := \tau\bX\big[\mbE\{ Z(\bbeta) \mid \bX\} - \tau\bX^\top\btheta\big]$. Under the joint regression model \eqref{jointqtesreg}, we have $\psi(\bbeta^*,\btheta^*;\bX)=0$. Detailed illustration of the motivation can be found at Section 3.2 in \cite{zhang2023hdesreg}.

Furthermore, \cite{zhang2023hdesreg} proposed a valid elementwise inference procedure for $\btheta^*$ based on the decorrelated score \citep{ning2017decorscor}. Suppose we are interested in $\theta_j$ for some $j\in\{1,\dots,p\}$. Let $X_{i,j}$ be the $j$-th element of $\bX_i$ and $\bX_{i,-j}$ denote the subvector of $\bX_i$ without $X_{i,j}$, the projection estimator is given by 
\begin{align}
    \label{offpendecorest}
    \hat{\bgamma}_j = \arg\min_{\bgamma\in\mbR^{p-1}} \frac{1}{2n}\sum_{i=1}^n \big(X_{i,j}-\bX_{i,-j}^\top\bgamma\big)^2 + \lambda_m \Vert\bgamma\Vert_1.  
\end{align}
Define the decorrelated score $S_n(\theta_j,\btheta_{-j},\bbeta,\bgamma)=n^{-1}\sum_{i=1}^n \{Z_i(\bbeta)-\tau X_{i,j}\theta_j -\tau\bX_{i,-j}^{\top}\btheta_{-j}\}(X_{i,j}-\bX_{i,-j}^\top\bgamma)$. The debiased estimator is then given as
\begin{align}
    \label{offdebiasest}
    \tilde{\theta}_j = \hat{\theta}_j-\frac{S_n(\hat{\theta}_j, \hat{\btheta}_{-j}, \hat{\bbeta}, \hat{\bgamma})}{\nabla_b S_n(b, \hat{\btheta}_{-j}, \hat{\bbeta}, \hat{\bgamma})\mid_{b=\hat{\theta}_j}}.
\end{align}
Given the consistent variance estimator, confidence intervals for $\theta_j$ can be constructed based on \eqref{offdebiasest}.

\subsection{Online Estimation}
Shift our focus to the online environment. Suppose we observe a total of $b$ batches of \textit{i.i.d.} data arriving sequentially, denoted as $\mcD_1, \dots, \mcD_b$. Define the batch size $n_s=\vert\mcD_s\vert$ and denote $\mcD_s=\{Y_{si}, \bX_{si}\}_{i=1}^{n_s}$ for $s=1,\dots,b$. Based on the data stream, we aim to construct renewable estimation and inference procedures for the ES regression coefficient $\btheta^*$ in \eqref{jointqtesreg}. Given storage limitations, each batch can be processed only once and will be omitted from the system. Thus, our procedure merely uses the current batch with condensed information from historical data, and is free of the full raw data. 

We now introduce the online estimation method. At the beginning, we observe $\mcD_1$ and apply the aforementioned offline procedures to obtain the estimators in \eqref{offpenquantest},\eqref{offpenESest}, \eqref{offpendecorest} and \eqref{offdebiasest}, denoted severally as $\hat{\bbeta}^{(1)},\ \hat{\btheta}^{(1)},\ \hat{\bgamma}^{(1)}_j$ and $\tilde{\theta}^{(1)}_j$. When the second batch $\mcD_2$ arrives after $\mcD_1$, we aim to update these estimators one by one.
Firstly, we get $\hat{\bbeta}^{(2)}$ using the online estimation for penalized smooth quantile regression proposed by \cite{Xie2023onlinesmoothquantile}. A major concern of their method is the difference from the original check loss in \eqref{offpenquantest} caused by the kernel smoothing. In fact, their estimator is capable of bounding the error with respect to $\bbeta^*$: the $\ell_2$-norm error is bounded by the order of $\sqrt{s_q\log(p)/N_2}$. According to Proposition \ref{propbetaesterr} and Corollary \ref{crlryesterr}, $\hat{\bbeta}^{(2)}$ satisfies the theoretical condition for controlling the error of the ES regression coefficient estimation. Therefore, it is reasonable to use the renewable estimator of \cite{Xie2023onlinesmoothquantile}.   

Next, we aim to get $\hat{\btheta}^{(2)}$. Denote the batch loss $L(\bbeta,\theta;\mcD_s) = (2n_s)^{-1}\sum_{i=1}^{n_s}\big\{ Z_{si}(\bbeta)-\tau\bX_{si}^\top\btheta \big\}^2$ and let $\lambda_e^{(s)}$ be the penalization parameter set at time $s$. Since the previous raw data $\mcD_1$ is deleted, reconstructing the loss $L(\bbeta,\theta;\mcD_1)$ becomes the central problem. Inspired by the smoothness of loss function $L$, we apply Taylor's expansion at $\hat{\btheta}^{(1)}$ as
\begin{align*}
    L(\hat{\bbeta}^{(1)},\btheta;\mcD_1) =& \frac{1}{2n_1}\sum_{i=1}^{n_1}\big\{ Z_{1i}(\hat{\bbeta}^{(1)})-\tau\bX_{1i}^\top\btheta \big\}^2 \\
    =& \frac{1}{2n_1}\sum_{i=1}^{n_1}\big\{ Z_{1i}(\hat{\bbeta}^{(1)})-\tau\bX_{1i}^\top\hat{\btheta}^{(1)} \big\}^2 -  \frac{\tau}{n_1}\sum_{i=1}^{n_1} \big\{ Z_{1i}(\hat{\bbeta}^{(1)})-\tau\bX_{1i}^\top\hat{\btheta}^{(1)} \big\}\bX_{1i}^\top \big(\btheta-\hat{\btheta}^{(1)}\big) \\
    &+ \frac{\tau^2}{2n_1} \big(\btheta-\hat{\btheta}^{(1)}\big)^\top \sum_{i=1}^{n_1} \bX_{1i}\bX_{1i}^\top \big(\btheta-\hat{\btheta}^{(1)}\big) + O\big(\Vert\btheta-\hat{\btheta}^{(1)}\Vert_2^2 \big) \\
    := & A_0 + A_1 + \frac{\tau^2}{2n_1} \big(\btheta-\hat{\btheta}^{(1)}\big)^\top \sum_{i=1}^{n_1} \bX_{1i}\bX_{1i}^\top \big(\btheta-\hat{\btheta}^{(1)}\big) + R_1. 
\end{align*}
Note that $A_0$ is a constant and could be ignored for optimization. Moreover, $\hat{\btheta}^{(1)}$ is the minimizer computed on $\mcD_1$. According to the KKT condition, $A_1=\lambda_e^{(1)}\kappa_1 (\btheta-\hat{\btheta}^{(1)})$, where $\kappa_1$ is a subgradient of $\Vert\cdot\Vert$ at $\hat{\btheta}^{(1)}$. By H{\"o}lder's inequality, $|A_1| \leq \lambda_e^{(1)} \Vert\kappa_1\Vert_{\infty} \Vert\btheta-\hat{\btheta}^{(1)}\Vert_1 = O\big(\lambda_e^{(1)}\Vert\btheta-\hat{\btheta}^{(1)}\Vert_1\big)$. Given $\lambda_e^{(1)}\asymp \sqrt{\log(p)/n_1}$, $A_1$ is then trivial relative to the quadratic term. Finally, the remainder $R_1$ is asymptotically negligible, thus we singly keep the quadratic term and employ the following surrogate loss function:
\begin{align}
    \label{D2surrloss}
    \mcL_{2}(\btheta) = \frac{1}{N_2} \Big[\frac{n_1}{2}\tau^2\big(\btheta-\hat{\btheta}^{(1)}\big)^\top\hat{\bSigma}_1\big(\btheta-\hat{\btheta}^{(1)}\big) + n_2 L(\hat{\bbeta}^{(2)},\btheta;\mcD_2) \Big]
\end{align}
where $\hat{\bSigma}_1=n_1^{-1}\sum_{i=1}^{n_1}\bX_{1i}\bX_{1i}^\top$. The renewable estimator for $\btheta^*$ is then given as 
\begin{align}
    \label{b2esregest}
    \hat{\btheta}^{(2)} = \arg\min_{\btheta\in\mbR^p} \mcL_2(\btheta) + \lambda_e^{(2)} \Vert\btheta\Vert_1.
\end{align}
The estimation in \eqref{b2esregest} only requires the current data $\mcD_2$ and the previous sample covariance $\hat{\bSigma}_1$. Analogously, at the time $\mcD_b$ is collected, the surrogate loss function is defined as 
\begin{align}
    \label{Dbsurrloss}
     \mcL_{b}(\btheta) = \frac{1}{N_b} \Big[\frac{N_{b-1}}{2} \tau^2\big(\btheta-\hat{\btheta}^{(b-1)}\big)^\top\tilde{\bSigma}_{b-1}\big(\btheta-\hat{\btheta}^{(b-1)}\big) + n_b L(\hat{\bbeta}^{(b)},\btheta;\mcD_b) \Big],
\end{align}
where $\tilde{\bSigma}_{b} = N_{b-1}^{-1}\sum_{s=1}^{b}n_s\hat{\bSigma}_{s}$ and $\hat{\bSigma}_{s}=n_s^{-1}\sum_{i=1}^{n_s}\bX_{si}\bX_{si}^\top$. 
The renewable estimator is then calculated as 
\begin{align}
    \label{bbesregest}
     \hat{\btheta}^{(b)} = \arg\min_{\btheta\in\mbR^p} \mcL_b(\btheta) + \lambda_e^{(b)} \Vert\btheta\Vert_1.
\end{align}

\subsection{Online inference}
We then introduce the online inference procedure for $\btheta^*$.
Denote the covariance matrices $\bSigma_{-j}=\mbE(\bX_{si,-j}\bX_{si,-j}^{\top})$, $\bSigma_{-j}^{(s)}=n_s^{-1}\sum_{i=1}^{n_s}\bX_{si,-j}\bX_{si,-j}^\top$ and $\tilde{\bSigma}_{-j}^{(b)} = N_{b}^{-1}\sum_{s=1}^{b}n_s\bSigma_{-j}^{(s)}$.
Firstly, consider the following model projecting the covariate of interest $X_{j}$ into the nuisance covariates $\bX_{-j}$: 
\begin{align}
    \label{modelprojreg}
    X_{si,j}=\bX_{si,-j}^{\top}\bgamma^*+\omega_{si}, \quad \mbE(\omega_{si}\bX_{si,-j})=0, \quad  \ &i=1,\dots,n_s, \notag \\ &s=1,\dots,b,
\end{align}
where $\bgamma^*=\arg\min_{\bgamma\in\mbR^{p-1}}\mbE(X_{j}-\bX_{-j}^{\top}\bgamma)^2=\bSigma_{-j}^{-1}\mbE(X_{j}\bX_{-j})$ under the condition that $\bSigma_{-j}$ is non-singular. For a time $s\in\{1,\dots,b\}$, the online projection lasso estimator is given by
\begin{align}
    \label{onlinegammaest}
    \hat{\bgamma}^{(s)} = \arg\min_{\bgamma\in\mbR^{(p-1)}} \frac{1}{2N_s}\sum_{t=1}^{s}\sum_{i=1}^{n_t}(X_{ti,j}-\bX_{ti,-j}^\top\bgamma)^2 + \lambda_m^{(s)}\Vert\bgamma\Vert_1.
\end{align}
To see why $\hat{\bgamma}^{(s)}$ is feasible, note that the objective function in \eqref{onlinegammaest} depends on data through the summary statistics
\begin{align*}
    T^{(s)}=\frac{1}{N_s}\sum_{t=1}^{s}\sum_{i=1}^{n_s}X_{ti,j}^2, \quad \mathbf{R}^{(s)}=\frac{1}{N_s}\sum_{t=1}^{s}\sum_{i=1}^{n_s}\bX_{ti,-j}\bX_{ti,-j}^{\top}, \quad \bm{U}^{(s)}= \frac{1}{N_s}\sum_{t=1}^{s}\sum_{i=1}^{n_s}\bX_{ti,-j}X_{ti,j}.
\end{align*}
The above statistics are extracted from $\tilde{\bSigma}_{s}$ directly as  sub-matrices, that is $T^{(s)}=(\tilde{\bSigma}_{s})_{j,j}$, $\mathbf{R}^{(s)}=(\tilde{\bSigma}_{s})_{-j,-j}$ and $\bm{U}^{(s)}=(\tilde{\bSigma}_{s})_{-j,j}$. As long as $\tilde{\bSigma}_{s}$ is in storage, these statistics are all available. Therefore, the estimator \eqref{onlinegammaest} is accessible in the online setting.

Moreover, define the score function on the $s$-th batch as
\begin{align*}
  S_{n_s}^{(s)}(\theta_j, \btheta_{-j}, \bbeta, \bgamma) = n_s^{-1}\sum_{i=1}^{n_s} \{Z_{si}(\bbeta)-\tau X_{si,j}\theta_j -\tau\bX_{si,-j}^{\top}\btheta_{-j}\}(X_{si,j}-\bX_{si,-j}^\top\bgamma)  
\end{align*}
and let $\hat{S}_{n_s}^{(s)}(\theta_j) = S_{n_s}^{(s)}\big({\theta}_j, \hat{\btheta}_{-j}^{(s)}, \hat{\bbeta}^{(s)}, \hat{\bgamma}^{(s)} \big)$. Further denote the derivative of $\hat{S}_{n_s}^{(s)}(\theta_j)$ as 
\begin{align}
    \label{derivscore}
    \hat{H}_{n_s}^{(s)}(\theta_j) = \partial_{\theta_j}\hat{S}_{n_s}^{(s)}(\theta_j) = -\frac{\tau}{n_s}\sum_{i=1}^{n_s}X_{si,j} \big(X_{si,j}-\bX_{si,-j}^\top\hat{\bgamma}^{(s)}\big). 
\end{align}
Since $\theta_j$ is excluded in the right hand of \eqref{derivscore}, we leave it out and simply write $\hat{H}_{n_s}^{(s)}$. Moreover, let $\tilde{H}_{N_b}=N_b^{-1}\sum_{s=1}^b n_s\hat{H}_{n_s}^{(s)}$ and
define the cumulative score function as $\tilde{S}_{N_b} = N_b^{-1} \sum_{s=1}^b n_s\hat{S}_{n_s}^{(s)}(\hat{\theta}_j^{(s)})$.
We then construct the following online debiased estimator: 
\begin{align}
    \label{onlinedebest}
    \tilde{\theta}_{j}^{(b)} =  \hat{\theta}_{j}^{(b)} - \tilde{S}_{N_b} / \tilde{H}_{N_b} + N_b^{-1}\sum_{s=1}^b n_s\hat{H}_{n_s}^{(s)}\big( \hat{\theta}_j^{(s)}-\hat{\theta}_j^{(b)} \big) / \tilde{H}_{N_b}.
\end{align}
Denote $W_{N_b}=N_b^{-1}\sum_{s=1}^b n_s \hat{H}_{n_s}^{(s)}\hat{\theta}_j^{(s)}$. To obtain the estimator in \eqref{onlinedebest}, only the current data $\mcD_b$ and the summary statistics $\big\{\hat{\theta}_j^{(b)}, \tilde{\bSigma}_b, \tilde{S}_{N_b}, \tilde{H}_{N_b}, W_{N_b}\big\}$ are involved. Thus, the proposed method is applicable for online updating. The required storage space is $O(p^2)$, which remains fixed as the new batch arrives. Compared to the offline debiased estimator \eqref{offdebiasest}, a second-order term is added on the right side of \eqref{onlinedebest}, which corrects the extra bias stemming from online renewal. 

Furthermore, we will show in Section \ref{sec:theory} that the proposed debiased estimator $\tilde{\theta}_{j}^{(b)}$ is asymptotically normal, thus valid confidence intervals could be constructed for $\theta_j^{*}$ coupled with the estimated variances. Specifically, it is demonstrated in Theorem \ref{theoremasympnorm} that $\sqrt{N_b}\tau(\tilde{\theta}_{j}^{(b)}-\theta_j^{*})$ converges in distribution to $N(0,\sigma_I^2/\sigma_{\omega}^4)$, where $\sigma_{\omega}^2=\mbE(\omega_{si}^2)$ and $\sigma_{I}^2 = \mbE\{\omega_{si}^2\Var(\varepsilon_{si,-}\mid\bX)\}$ with  $\varepsilon_{si}=Y_{si}-\bX_{si}^\top\bbeta^*$ and $\varepsilon_{si,-}=\min(\varepsilon_{si},0)$. We then establish consistent estimators of $\sigma_I^2$ and $\sigma_{\omega}^2$ and construct Wald-type confidence intervals. Let $\hat{S}_q^{(s)}$, $\hat{S}_e^{(s)}$ and $\hat{S}_m^{(s)}$ be the non-zero sets of estimators $\hat{\bbeta}^{(s)}$, $\hat{\btheta}^{(s)}$ and $\hat{\bgamma}^{(s)}$ with cardinalities $\hat{s}_q^{(s)}$, $\hat{s}_e^{(s)}$ and $\hat{s}_m^{(s)}$, respectively. Due to the $\ell_1$-norm penalty, the estimated sparsities $\hat{s}_q^{(s)}$, $\hat{s}_e^{(s)}$, $\hat{s}_m^{(s)}$ are fairly small, so we reasonably assume that $N_s > \hat{s}_q^{(s)}+\hat{s}_e^{(s)}+\hat{s}_m^{(s)}$. Then define the following estimators of $\sigma_I^2$ and $\sigma_{\omega}^2$:
\begin{align}
    \label{varest}
    \hat{\sigma}_{I,(b)}^2 = \frac{1}{N_b-\hat{s}_q^{(b)}-\hat{s}_e^{(b)}-\hat{s}_m^{(b)}} \sum_{s=1}^{b}\sum_{i=1}^{n_s}\hat{\omega}_{si}^2\big\{e_{si}(\hat{\bbeta}^{(s)})\big\}^2, \quad \hat{\sigma}_{\omega,(b)}^2 = \frac{1}{N_b-\hat{s}_m^{(b)}}\sum_{s=1}^{b}\sum_{i=1}^{n_s}\hat{\omega}_{si}^2,
\end{align}
where $\hat{\omega}_{si}=X_{si,j}-\bX_{si,-j}^\top\hat{\bgamma}^{(s)}$ and $e_{si}(\bbeta)=Z_{si}(\bbeta)-\tau\bX_{si}^\top\btheta^*$. We will show in Theorem \ref{theoremconsvarest} that $\hat{\sigma}_{I,(b)}^2$ and $\hat{\sigma}_{\omega,(b)}^2$ are both consistent under certain restrictions on $(n_1,N_b,p,s_m,s_r,s_0)$. The two estimators in \eqref{varest} are online renewable by recording summary statistics $\{\hat{\omega}_{si}, \ e_{si}(\hat{\bbeta}^{(s)})\}$, thus the confidence interval for $\theta_j^*$ at time point $b$ could be given as
\begin{align}
    \label{onlineconfintv}
    \bigg[\tilde{\theta}_{j}^{(b)}-\frac{\hat{\sigma}_{I,(b)}}{\sqrt{N_b}\tau\Phi(1-\alpha/2)\hat{\sigma}_{\omega,(b)}^2} ,\ \tilde{\theta}_{j}^{(b)}+\frac{\hat{\sigma}_{I,(b)}}{\sqrt{N_b}\tau\Phi(1-\alpha/2)\hat{\sigma}_{\omega,(b)}^2}\bigg].
\end{align}

\section{Theory}
In this section, we present a detailed theoretical analysis of the methods discussed in the previous section. Non-asymptotic error bounds are provided for the online estimator in \eqref{bbesregest} and asymptotic normality is established for the online debiased estimator \eqref{onlinedebest}.
\label{sec:theory}
\subsection{Estimation error bounds}
Let $\bSigma=\mbE(\bX\bX^\top)$ be the population covariance matrix and $\sigma_{\bX}=\max_{j\in\{1,\dots,p\}}\Sigma_{j,j}$. Let $\varepsilon = Y-\bX^\top\bbeta^*$ be the residual of the linear quantile regression model in \eqref{jointqtesreg} and $\varepsilon_{-} = \min (\varepsilon,0)$ be its negative part. Define the sparsity level $s_q=\Vert\bbeta^*\Vert_0$, $s_e=\Vert\btheta^*\Vert_0$, $s_m=\Vert\bgamma^*\Vert_0$ and $s_0=\max\{s_q,s_e\}$. Moreover, let $\bm{\Omega}=\bSigma^{-1}$ be the precision matrix and define the sparse set $S_r=\{k\neq r:\bm{\Omega}_{k,r}\neq 0\}$ with cardinality $s_r= \#S_r$. Define ball sets $\mbB_{\bSigma}(r_0)=\{\bu\in\mbR^p:\Vert \bu\Vert_{\bSigma} \leq r_0\}$ and $\mbB_{1}(r_1)=\{\bu\in\mbR^p:\Vert \bu\Vert_{1} \leq r_1\}$. Let $\sup_{\Vert \bu\Vert_2=1} \mbE\{ (\bX^\top \bu)^3\} / [\mbE\{ (\bX^\top \bu)^2\}]^{3/2} = M_3$. For any index $S\subseteq \{1,\dots,p\}$, let $\bdelta_{S}$ be the subvector of $\bdelta$ with elements in $S$. Moreover, let $\mbC(S)=\{\bdelta\in\mbR^p:\Vert\bdelta_{S^c}\Vert_1\leq3\Vert\bdelta_{S}\Vert_1\}$ be an $\ell_1$-cone. The following technical assumptions are made. 

\begin{assumption}
    \label{assumpepsilon}
     The conditional cumulative distribution function $F_{\varepsilon\mid \bX}(\cdot)$ of $\varepsilon$ given $\bX$ is continuous and differentiable, satisfying $|F_{\varepsilon\mid\bX}(t)-F_{\varepsilon\mid\bX}(0)| \leq F_u|t|, \ \forall t\in\mbR$ for some positive constant $F_u$. Moreover, there exist some constants $q_{\varepsilon} \geq \sigma_{\varepsilon}>0$ that $\mbE(\varepsilon_{-}^2\mid \bX) \leq \sigma_{\varepsilon}^2$ and $\mbE(|\varepsilon_{-}|^k\mid\bX) \leq k!\sigma_{\varepsilon}^2 q_{\varepsilon}^{k-2}/2$ for any integer $k\geq 3$.     
\end{assumption}

\begin{assumption}
    \label{assumpxboundRME}
    Assume that there exists some constant $B_{\bX},\ M_4\geq 1$ such that $\Vert\bX\Vert_{\infty} \leq B_{\bx}$ and $\sup_{\Vert \bu\Vert_2=1} \mbE\{ (\bX^\top \bu)^4\} / [\mbE\{ (\bX^\top \bu)^2\}]^{2} \leq M_4$. Moreover, define the following restricted minimum eigenvalue (RME)
    \begin{align*}
        \rho^2 = \inf_{S\subset\{1,\dots,p\}:|S|\leq s}\ \inf_{\bdelta\in\mbC(S)}\frac{\bdelta^\top\bSigma\bdelta}{\Vert\bdelta_{S}\Vert_2^2}>0.
    \end{align*}     
\end{assumption}

Assumption \ref{assumpepsilon} indicates a Lipschitz-continuous conditional distribution function and imposes a Bernstein-type condition for $\varepsilon_{-}$. Assumption \ref{assumpxboundRME} is mild on covariates $\bX$, for example, it is satisfied when $\bX$ is sub-Gaussian. Moreover, assumption \ref{assumpxboundRME} implies that the covariance matrix $\bSigma$ has a positive bounded restricted minimal eigenvalue, which is common in high-dimensional data \citep{martinHDS2019}. Both assumptions are necessary for theoretical analysis of ES regression and are similarly made in \cite{zhang2023hdesreg}.    

\begin{theorem}[Estimation error bounds]
    \label{thmesterrbound}
    Under Assumption \ref{assumpepsilon}-\ref{assumpxboundRME}, for any $s\in\{1,\dots,b\}$, suppose $\hat{\bbeta}^{(s)}\in\bbeta^*+\mbB_{\bSigma}(r_0)\cap\mbB_{1}(r_1)$ for some $r_0,r_1>0$. Given the penalized parameters $\lambda_e^{(s)}\geq\max[4\sigma_{\bX}(2q_{\epsilon}+5B_{\bX}r_1)\{\log(p)/N_s\}^{1/2} ,F_uM_3\rho s_0^{-1/2}r_0^2]$, $\lambda_e^{(s)}\sqrt{N_s}\asymp  \lambda_e^{(s+1)}\sqrt{N_{s+1}}$ and the first sample size
    $n_1 > C\max[M_4\sigma_{\bX}^2\rho^{-2}(s_0+2)\log(2p),\ 2q_{\epsilon}^2B_{\bX}^2\sigma_{\epsilon}^{-2}\sigma_{\bX}^{-2}\log(2p)]$ for some large constant $C>1$, the online $\ell_1$-penalized ES regression estimator in \eqref{bbesregest} satisfies
    \begin{align*}
       \tau\Vert\hat{\btheta}^{(s)} - \btheta^* \Vert_{\bSigma} \leq \frac{4}{\rho} s_e^{1/2} \lambda_e^{(s)}, \quad \tau\Vert\hat{\btheta}^{(s)} - \btheta^* \Vert_1 \leq \frac{52}{\rho^2} s_e \lambda_e^{(s)}
    \end{align*}
    with probability at least $1-5/p$. 
\end{theorem}
Theorem \ref{thmesterrbound} establishes the non-asymptotic error bounds for the online expected shortfall regression estimator $\hat{\btheta}^{(s)}$, provided that the sample size of the first batch satisfies a growth condition. The convergence rates of $\hat{\btheta}^{(s)}$ under $\ell_{\bSigma}$- and $\ell_1$- norms are associated with $s_e,\lambda_e^{(s)}$ and constant $\rho$. Furthermore, Theorem \ref{thmesterrbound} shows the impact of quantile regression coefficient estimation on the ES regression coefficient estimation in the online setting. It implies that error bounds of the online ES regression estimator $\hat{\btheta}^{(s)}$ are conditioned on the accuracy of the quantile regression estimator plugged in, that is, $\Vert\hat{\bbeta}^{(s)}-\bbeta^*\Vert_{\bSigma}\leq r_0$ and $\Vert\hat{\bbeta}^{(s)}-\bbeta^*\Vert_1\leq r_1$. To specify this prerequisite, we next provides a precise expression for the radius $r_0$ and $r_1$ with respect to $(N_s,p,s_q)$. 

\begin{assumption}
    \label{assumpepsconddense}
    The conditional density function $f_{\varepsilon\mid \bX}(\cdot)$ of $\varepsilon$ given $\bX$  exists and is continuous. Moreover, there exist constants $f_l,\ f_u \ ,f_0>0$ such that $f_l \leq f_{\varepsilon\mid \bX}(0) \leq f_u$ and $|f_{\varepsilon\mid \bX}(t)-f_{\varepsilon\mid \bX}(0)| \leq f_0 |t|$. In addition, suppose $\bSigma$ is positive definite with bounded eigenvalues $\Lambda_1\geq\dots\geq\Lambda_p>0$.
\end{assumption}
\begin{assumption}
    \label{assumpqtestkernel}
    The kernel function $K(\cdot)$ is twice differentiable with bounded first and second derivatives and satisfies: (i) $K(u)\geq 0$ and $K(-u)=K(u)$ for all $u\in\mbR$; (ii) $\int_{-\infty}^{\infty}K(u){\rm d}u=1$; (iii) $\int_{-\infty}^{\infty} u^2K(u){\rm d}u<\infty$; (iv). For $k\geq 1$, denote $\kappa_k=\int_{-\infty}^{\infty}|u|^kK(u){\rm d}u$, $\kappa_l=\min_{|u|\leq 1}K(u)>0$ and $\kappa_u=\sup_{u\in\mbR}K(u) < \infty$.
\end{assumption}

\begin{proposition}
    \label{propbetaesterr}
    Under Assumption \ref{assumpepsconddense} and \ref{assumpqtestkernel}, for any $s=1,\dots,b$, choose the tuning parameter as $\lambda_q^{(s)}=O\big(\sigma_{\bX}\sqrt{\tau(1-\tau)\log(p)/N_s}\big)$ and the bandwidth $h_s\asymp\{s_q\log(p)/N_s\}^{1/4}$. Provided that $n_1\geq\{s_q\log(p)\}^9$, the online smoothed quantile regression estimator $\hat{\bbeta}^{(s)}$ satisfies the error bounds
    \begin{align*}
        \Vert\hat{\bbeta}^{(s)}-\bbeta^*\Vert_1 \leq C_{1,s}s_q\{\log(p)/N_s\}^{1/2}, \quad \Vert\hat{\bbeta}^{(s)}-\bbeta^*\Vert_{\bSigma} \leq C_{0,s}\{s_q\log(p)/N_s\}^{1/2}
    \end{align*}
    with probability at least $1-4/p$, for some positive constants $C_{1,s}, C_{0,s}$ relying on $\sigma_{\bX}$, $B_{\bX}$, $\Lambda_1$, $\Lambda_p$, $f_l$, $f_0$, $f_u$, $\kappa_l$, $\kappa_2$.
\end{proposition}
Proposition \ref{propbetaesterr} follows directly from Theorem 2 in \cite{Xie2023onlinesmoothquantile} by taking $\gamma=3/4$. It states that $\hat{\bbeta}^{(s)}$ satisfies the conditions in Theorem \ref{thmesterrbound} with $r_0\asymp\{s_q\log(p)/N_s\}^{1/2}$ and $r_1\asymp s_q\{\log(p)/N_s\}^{1/2}$. Combining Theorem \ref{thmesterrbound} and Proposition \ref{propbetaesterr}, we then derive explicit error bounds for $\hat{\btheta}^{(s)}$ relying on $(N_s, p, s_0)$. 

\begin{corollary}
    \label{crlryesterr}
    Under Assumption \ref{assumpepsilon}-\ref{assumpqtestkernel}, for any $s\in\{1,\dots,b\}$, let the penalized parameters $\lambda_q^{(s)},\lambda_e^{(s)}\asymp(\sqrt{\log(p)/N_s})$, the bandwidth $h_s\asymp\{s_q\log(p)/N_s\}^{1/4}$ and the first sample size $n_1\gtrsim \{s_0\log(p)\}^9$. We then have $\Vert\hat{\btheta}^{(s)}-\btheta^*\Vert_{\bSigma}\lesssim \{s_e\log(p)/N_s\}^{1/2}$ and $\Vert\hat{\btheta}^{(s)}-\btheta^*\Vert_1\lesssim s_e\{\log(p)/N_s\}^{1/2}$ with probability at least $1-9/p$.
\end{corollary}
As $s$ grows from $1$ to $b$, the $\ell_1$- and $\ell_{\bSigma}$-norm error bounds of $\hat{\btheta}^{(s)}$ become tighter and eventually achieve the orders of $O(\{s_e\log(p)/N_b\}^{1/2})$ and $O(s_e\{\log(p)/N_b\}^{1/2})$, respectively. The final orders align with those of the offline ES regression estimator with full sample, as demonstrated by Theorem 1 in \cite{zhang2023hdesreg}. 

\subsection{Asymptotic normality}
We next investigate the asymptotic normality of the online debiased ES regression estimator $\tilde{\theta}_j$ and demonstrate the consistency of variance estimation. These results further require the error bounds for the projection coefficient estimator $\hat{\bgamma}^{(s)}$, as presented in Proposition \ref{propgammaesterr}.   

\begin{assumption}
    \label{assumpomega}
    Assume the projection residual $\vert\omega_i^{(s)}\vert \leq B_{\omega}$ for some constant $B_{\omega}>0$ independent of $(n_s,p,s_r)$. Moreover, assume the error term $\varepsilon$ is sub-exponential. In addition, the largest diagonal element of $\bSigma$ satisfies $\max_{j}\bSigma_{j,j}=O(1)$.
\end{assumption}

\begin{proposition}
    \label{propgammaesterr}
    Under Assumption \ref{assumpepsilon}, \ref{assumpxboundRME} and \ref{assumpomega}, suppose that $\lambda_m^{(s)}\asymp\sqrt{\log(p)/N_s}$ and $n_1 \gtrsim s_r\log(p)$, then the estimator $\hat{\bgamma}^{(s)}$ in \eqref{onlinegammaest} satisfies the error bounds
    \begin{align*}
        \Vert\hat{\bgamma}^{(s)}-\bgamma^*\Vert_{\bSigma_{-j}} \leq 3\phi_s^{-1}s_r^{1/2}\lambda_m^{(s)}, \quad \Vert\hat{\bgamma}^{(s)}-\bgamma^*\Vert_{1} \leq 12\phi_s^{-2}s_r\lambda_m^{(s)}
    \end{align*}
    with probability at least $1-p^{-3}$. Here $\phi_s$ is a constant such that for all $\bdelta\in\mbR^{p-1}$ satisfying $\Vert\bdelta_{S_r^c}\Vert_1 \leq 3\Vert\bdelta_{S_r}\Vert_1$, it holds that $\Vert\bdelta_{S_r}\Vert_1\leq\phi_s^{-1}(s_r\bdelta^\top\tilde{\bSigma}_{-j}^{(s)}\bdelta)^{1/2}$. 
\end{proposition}
Proposition \ref{propgammaesterr} follows directly from Lemma 1 in \cite{han2021delasso}. It guarantees the estimation accuracy of the online projection regression estimator $\hat{\bgamma}^{(s)}$. Up to now, we have exhibited the error bounds for the online estimators $\hat{\bbeta}^{(s)},\ \hat{\btheta}^{(s)},\ \hat{\bgamma}^{(s)}$ in Proposition \ref{propbetaesterr}, Corollary \ref{crlryesterr} and Proposition \ref{propgammaesterr}, respectively. Next, we establish the asymptotic normality for the online debiased estimators \eqref{onlinedebest} in the following Theorem \ref{theoremasympnorm}.  

\begin{theorem}[Asymptotic distribution]
    \label{theoremasympnorm}
     Assume the conditions in Theorem \ref{thmesterrbound}, Proposition \ref{propbetaesterr} and Proposition \ref{propgammaesterr} hold. Suppose the regularized parameters $\lambda_e^{(b)},\lambda_q^{(b)}, \lambda_m^{(b)}\asymp\sqrt{\log(p)/N_b}$, the batch size $n_s\geq(q_{\epsilon}/\sigma_{\epsilon})^2(B_{\bX}/\sigma_{\bX})^2\log(p)$ for $s=1,\dots,b$ and $\max(s_0^2,s_r^2) s_m \log^2(p) b \{1+\log(N_b/n_1)\} = o(N_b)$. The online debiased estimator $\tilde{\theta}_j^{(b)}$ defined in \eqref{onlinedebest} then satisfies $\tau \sqrt{N_b}(\tilde{\theta}_j^{(b)}-\theta_j^*) \stackrel{d}{\rightarrow} N(0,\sigma_{I}^2/\sigma_{\omega}^4)$ as $(N_b,p) \to \infty$.
\end{theorem}

At the end of this section, we demonstrate the consistency of the variance estimators in \eqref{varest}, which supports our empirical confidence interval \eqref{onlineconfintv}.
\begin{theorem}[Consistency of variance estimation]
    \label{theoremconsvarest}
    Under conditions in Theorem \ref{thmesterrbound}, Proposition \ref{propbetaesterr} and Proposition \ref{propgammaesterr}, suppose that $\Vert\bgamma^*\Vert_1$ is bounded. With regularized parameters $\lambda_e^{(b)},\lambda_q^{(b)}, \lambda_m^{(b)}\asymp\sqrt{\log(p)/N_b}$ and the condition $\max(s_0^2,s_r)\log(p)=o(n_1)$, the online variance estimators $\hat{\sigma}_{I,(b)}^2$ and $\hat{\sigma}_{\omega,(b)}^2$ defined in \eqref{varest} are consistent, that is, $\hat{\sigma}_{I,(b)}^2 \stackrel{p}{\rightarrow} \sigma_{I}^2$ and $\hat{\sigma}_{\omega,(b)}^2 \stackrel{p}{\rightarrow} \sigma_{\omega}^2$ as $N_b\to\infty$. In addition, under the growth condition $s_m^2p^2=o(N_b)$, we have $\sqrt{N_b}(\hat{\sigma}_{\omega,(b)}^2-\sigma_{\omega}^2) \stackrel{d}{\rightarrow} N(0, \mbE\omega_{si}^4-\sigma_{\omega}^4)$ as $N_b\to\infty$.
\end{theorem}

\section{Numerical studies}
In this section, we implemented simulations to examine the empirical performances of the proposed methods. In addition, our methods are conducted on a real data example. Both the renewable estimation and inference methods are evaluated.

\subsection{Tuning parameter selection and algorithms}
To begin with, we discuss the selection of the tuning parameter and the optimization algorithms. Our renewable procedure includes three regressions on each batch: the quantile regression, the ES regression \eqref{bbesregest}, and the projected regression \eqref{onlinegammaest}. Existing literature has provided feasible algorithms for the renewable quantile regression \citep{Xie2023onlinesmoothquantile} and the projected regression \citep{han2021delasso}, thus we mainly focus on the ES regression \eqref{bbesregest}.

Guided by the theoretical order in Theorem \ref{thmesterrbound}, we set $\lambda_e^{(b)}=c_e\sqrt{\log(p)/N_b}$ and tune the regularization parameter $\lambda_e$ adaptively through cross-validation, coupled with the rolling-original-recalibration (ROR) scheme. Specifically, suppose we already have $\{\hat{\btheta}_{b-1}(c):c\in \mathcal{C}_e\}$, where $\hat{\btheta}_{b-1}(c)$ represents estimators on $\mcD_{b-1}$ produced under $\lambda_e^{(b-1)}=c\sqrt{\log(p)/N_{b-1}}$. For the $b$-th batch, we select $c_e$ from a candidate set $\mathcal{C}_e$ as: 
\begin{align}
    \label{lambdaROR}
    c_e^{(b)} = \underset{c \in \mathcal{C}_e}{\arg\min} \ L\big\{\widehat{\btheta}_{b-1}(c) ; \mathcal{D}_b\big\}.   
\end{align}
The above ROR scheme has recently been adopted by \cite{Xie2023onlinesmoothquantile,rao2025onlinesvm} and shown to be stable and efficient for renewable regularized estimation. 


Moreover, we adopt the proximal gradient descent (PGD) algorithm to solve penalized estimation \eqref{bbesregest}. Define 
the gradient of $\mcL_b$ as
\[
\nabla\mathcal{L}_b(\btheta):=\frac{1}{N_b} \Big[\tau^2 (N_{b-1} \tilde{\bSigma}_{b}+n_b\hat{\bSigma}_b)\big(\btheta-\hat{\btheta}^{(b-1)}\big)-\tau\sum_{i=1}^{n_b}\big\{Z_{bi}(\hat{\bbeta}^{(b)})- \tau\bX_{bi}^\top\hat{\btheta}^{(b-1)}\big\}\ \bX_{bi} \Big]
\] 
Given initial value $\hat{\btheta}^{(b,0)}$, the $t$-th time of PGD algorithm computes:
\begin{align}
	\label{APGprox}
	\hat{\btheta}^{(b,t)}= \textbf{prox}_{\eta_{t}\lambda_e^{(b)}\Vert\cdot\Vert_1}
	\left\{\hat{\btheta}^{(b,t-1)}-\eta_{t} \nabla\mcL_b(\hat{\btheta}^{(b,t-1)}) \right\}, \quad t\geq 1, 
\end{align}
where $\eta_{t}$ is the step size and $\textbf{prox}_{\eta_t\lambda_b\Vert\cdot\Vert_1}(u)=\textbf{sgn}(u)(|u|-\eta_t\lambda_b)$ is the soft-thresholding operator.

\subsection{Simulation: renewable estimation}
We generate a data stream $\{\mcD_1,\dots,\mcD_b\}$ from the following heteroscedastic linear model
\begin{align}
    \label{modelsimu}
    Y_{si} = \bX_{si}^{\top}\bxi^* + (\bX_{si}^{\top}\bzeta^*)\epsilon_{si}, \quad i=1,\dots,n_s;\ s=1,\dots,b.
\end{align}
Here $\{\epsilon_{si},i=1,\dots,n_s;\ t=1,\dots,s\}$ are \textit{i.i.d.} random errors genereated from standard normal distribution and $\bX_{si}$ is the covariates under the following designs: 

Covariate 1. $X_{si,j}=|z_j|$ for $j=1,\dots,p$ with $z\sim N_p(0,I_p)$;

Covariate 2. $X_{si,j}=|z_j|$ for $j=1,\dots,p$ with $z\sim N_p(0,\Sigma)$ and $\Sigma_{k,l}=0.8^{|k-l|}$;

Covariate 3. $X_{si,j}$ are i.i.d. distributed from uniform distribution $U(0,1.5)$. 

Here $\bxi^*=(\xi_1^*,\dots,\xi_{s^*}^*,0,\dots,0)\in\mbR^p$ and $\bzeta^*=(\zeta_1^*,\dots,\zeta_{s^*/2}^*,0,\dots,0)\in\mbR^p$ are $p$-dimensional coefficients with sparsity level $s^*$. They are specified as $\xi_j^*=2$ for $j=1,\dots,s^*/2$, $\xi_j^*=1$ for $j=s^*/2+1,\dots,s_0$ and $\zeta_j^*=1/3$ for $j=1,\dots,s^*/2$. Under model \eqref{modelsimu}, the true parameters of quantile regression and expected shortfall satisfy 
\begin{align*}
    \bbeta^* = \bxi^* + \bzeta^*Q_{\tau}(\epsilon), \quad \btheta^* = \bxi^* + \bzeta^*ES_{\tau}(\epsilon).
\end{align*}
Throughout our studies, we set $s^*=10, p=1000$ and select quantile level $\tau$ from $\{0.05,0.1,0.2\}$. We consider a setting with equal batch sizes, that is $n_s=N_b/b$ for $s=1,\dots,b$, and we set $b=1/\tau$ and $n_s=300,500$.  

We compare the proposed method with the offline benchmark using full data, which is the high-dimensional ES regression estimator proposed by \cite{zhang2023hdesreg}, shortened as ``off-l1-ES". For both methods, we report the following metrics: 
(i). Error-S: $\Vert\hat{\btheta}_{S_e}-\btheta^*_{S_e}\Vert_2$, the $\ell_2$-norm error on the signal set $S_e$;
(ii). Error-N: $\Vert\hat{\btheta}_{S_e^c}\Vert_2$, the $\ell_2$-norm error on the noise set $S_e^c$;
(iii). False positive rate: $\# \{ \hat{S}_e \cap S_e^c\} / \# S_e^c$, the proportion of elements in the true noise set $S_e^c$ that are estimated as nonzero.  

The results are summarized in Table \ref{table1}. We observe that the Errors and FPRs of our proposed renewable $\ell_1$-penalized ES regression estimator are slightly higher than those of the offline methods. This suggests that our methods, even under storage limits, achieve almost the same accuracy as the offline benchmark with full samples. Moreover, we display the $\ell_2$-errors of the proposed renewable estimator on every batch under Covariate 1 in Figure \ref{fig1}. It is shown that the errors decline in a curving manner as the batches arrive. Eventually, it reaches a level that is fairly close to the offline benchmark. A similar pattern is detected in the errors under Covariates 2 and 3. Their results are put in Figures \ref{fig2} and \ref{fig3} in  Appendix \ref{app:numres} due to the space constraint.     

\begin{table}[htbp]
\renewcommand\arraystretch{1.1}
\centering \tabcolsep 8pt 
\LTcapwidth 6in
\caption{Empirical estimation results of online-l1-ES and offline-l1-ES. }
\label{table1}

\begin{threeparttable}
\begin{tabular}{cccccccc}
\toprule
\multicolumn{1}{l}{} & \multicolumn{1}{l}{} & \multicolumn{3}{c}{$(n_t=300,b=1/\tau)$} & \multicolumn{3}{c}{$(n_t=500,b=1/\tau)$} \\ \midrule
\multicolumn{8}{c}{Covariate 1}                                                          \\ \hline
$\tau$  & method    & Error-S & Error-N & FPR   & Error-S & Error-N & FPR   \\ \hline
\multirow{2}{*}{0.05} & on-l1-ES  & 0.566   & 0.248   & 0.046 & 0.443   & 0.225   & 0.042 \\
 & off-l1-ES & 0.487   & 0.221   & 0.044 & 0.387   & 0.173   & 0.043 \\ \hline
\multirow{2}{*}{0.1}  & on-l1-ES  & 0.613   & 0.289   & 0.045 & 0.496   & 0.258   & 0.046 \\
 & off-l1-ES & 0.546   & 0.24    & 0.043 & 0.426   & 0.198   & 0.045 \\ \hline
\multirow{2}{*}{0.2}  & on-l1-ES  & 0.652   & 0.356   & 0.042 & 0.553   & 0.283   & 0.041 \\
 & off-l1-ES & 0.606   & 0.295   & 0.047 & 0.483   & 0.221   & 0.043 \\ \midrule
\multicolumn{8}{c}{Covariate 2}                                                          \\ \hline
$\tau$  & method    & Error-S & Error-N & FPR   & Error-S & Error-N & FPR   \\ \hline
\multirow{2}{*}{0.05} & on-l1-ES  & 0.429   & 0.158   & 0.02  & 0.334   & 0.112   & 0.02  \\
& off-l1-ES & 0.395   & 0.137   & 0.018 & 0.314   & 0.107   & 0.019 \\ \hline
\multirow{2}{*}{0.1}  & on-l1-ES  & 0.456   & 0.162   & 0.021 & 0.348   & 0.12    & 0.02  \\
& off-l1-ES & 0.434   & 0.155   & 0.019 & 0.341   & 0.12    & 0.017 \\ \hline
\multirow{2}{*}{0.2}  & on-l1-ES  & 0.509   & 0.167   & 0.02  & 0.422   & 0.113   & 0.02  \\
& off-l1-ES & 0.495   & 0.172   & 0.018 & 0.393   & 0.073   & 0.018 \\ \midrule
\multicolumn{8}{c}{Covariate 3}                                                          \\ \hline
$\tau$  & method    & Error-S & Error-N & FPR   & Error-S & Error-N & FPR   \\ \hline
\multirow{2}{*}{0.05} & on-l1-ES  & 0.642   & 0.307   & 0.04  & 0.513   & 0.228   & 0.04  \\
 & off-l1-ES & 0.62    & 0.284   & 0.044 & 0.494   & 0.225   & 0.044 \\ \hline
\multirow{2}{*}{0.1}  & on-l1-ES  & 0.703   & 0.368   & 0.041 & 0.556    & 0.269   & 0.04  \\
& off-l1-ES & 0.686   & 0.325   & 0.046 & 0.539   & 0.248   & 0.043 \\ \hline
\multirow{2}{*}{0.2}  & on-l1-ES  & 0.765   & 0.235   & 0.042 & 0.597   & 0.244   & 0.041 \\
 & off-l1-ES & 0.77    & 0.196   & 0.048 & 0.61    & 0.188   & 0.046 \\ \bottomrule
\end{tabular}


\end{threeparttable}
\end{table}

\begin{figure}
    \centering   \includegraphics[width=0.9\linewidth]{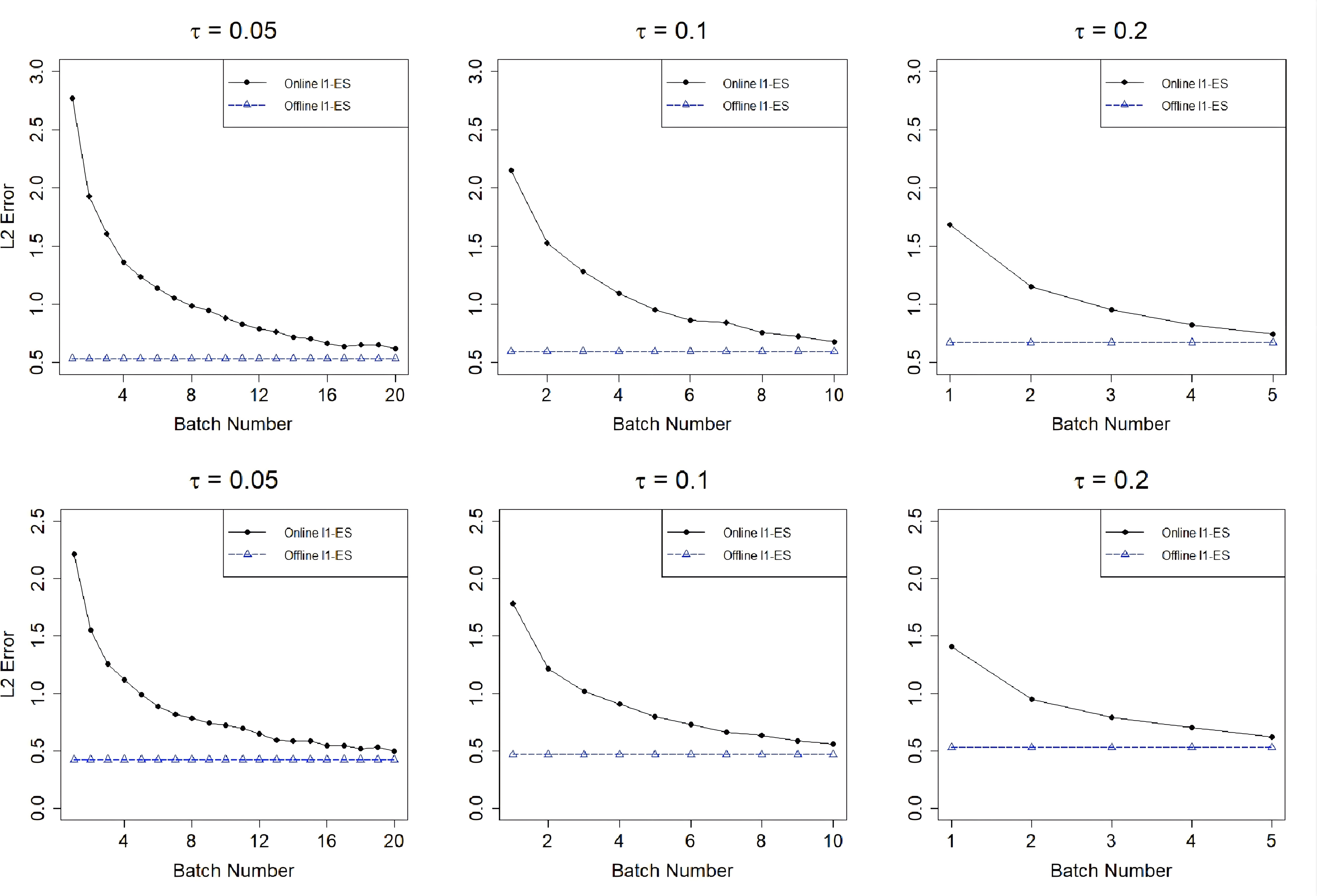}
    \caption{Estimation $\ell_2$- errors on each batch under Covariate 1}
    \label{fig1}
\end{figure}


\subsection{Simulation: renewable inference}
Subsequent to the renewable estimator in Section 4.2, we compute the online debiased estimator $\tilde{\theta}_j^{(b)}$ in \eqref{onlinedebest} and evaluate the performance of statistical inference by its associated confidence interval in \eqref{onlineconfintv}. On each batch $\mcD_s$, we obtained the projection regression estimator $\hat{\bgamma}^{(s)}$ in \eqref{modelprojreg} using the R package {\textit{glmnet}}, with the tuning parameter $\lambda_m^{(s)}$ selected by the ROR scheme. Provided the estimators $\hat{\bbeta}^{(s)},\hat{\btheta}^{(s)},\hat{\bgamma}^{(s)}$, we calculate the debiased estimator $\tilde{\theta}_j^{(s)}$ in \eqref{onlinedebest} and the variance estimators $\hat{\sigma}_{I,(s)}^2, \hat{\sigma}_{\omega,(s)}^2$ in \eqref{varest}. The corresponding confidence interval is then computed by \eqref{onlineconfintv}. 

The data is generated under the same design as in Section 4.2. Without loss of generality, let $\theta_1^*$ be the parameter of interest for statistical inference. In addition to the proposed renewable debiased estimator $\tilde{\theta}_1^{(s)}$, we also report the offline debiased estimator in \cite{zhang2023hdesreg} for comparison. The experiments are conducted under Covariate 1 and Covariate 2 with $s_0=10,p=1000$ and $n_s=500, b=1/\tau$. The lengths and coverage probabilities of the confidence intervals under two methods are averaged over 500 replications. Table \ref{table2} summarizes the averaged lengths and probabilities across multiple designs. The performance of our online method is more conservative than the offline method, as implied by lower coverage probabilities and shorter lengths of intervals. This issue might be due to the different choices of tuning parameters under online settings, led by more complex selection strategies. Despite the slight differences, the accuracy of the proposed inference procedure is relatively close to the offline benchmark.      

\begin{table}[htbp]
\renewcommand\arraystretch{1.1}
\centering \tabcolsep 10pt 
\LTcapwidth 6in
\caption{Empirical inference results of online-l1-ES and offline-l1-ES. }
\label{table2}

\begin{threeparttable}
\begin{tabular}{ccccc}
\toprule
\multicolumn{1}{l}{} & \multicolumn{4}{c}{Coverage Probability/Length}                   \\ \midrule
\multicolumn{1}{l}{} & \multicolumn{2}{c}{Covariate 1} & \multicolumn{2}{c}{Covariate 2} \\ \hline
$\tau$                  & on-l1-ES       & off-l1-ES      & on-l1-ES       & off-l1-ES      \\ \hline
0.05                 & 0.90/0.17      & 0.95/0.20      & 0.91/0.22      & 0.94/0.26      \\
0.1                  & 0.92/0.19      & 0.94/0.21      & 0.94/0.25      & 0.97/0.28      \\
0.2                  & 0.91/0.22      & 0.95/0.23      & 0.93/0.27      & 0.95/0.30      \\ 
\bottomrule
\end{tabular}


\end{threeparttable}
\end{table}

\subsection{Application: car insurance claim dataset} 
We further apply our method to investigate the effects of various factors on the expected shortfall of claim amounts of car insurance. Insurance claims typically exhibit heavy-tailed and heterogeneous characteristics, reflecting the diversity of risks. In this way, the expected shortfall becomes a core indicator in the car insurance business due to its ability to measure high-risk claims in the tail. Moreover, insurance companies may encounter the newly added application data on a monthly or quarterly basis, which calls for a renewable procedure. 

The data is available on Kaggle at \url{https://www.kaggle.com/datasets/xiaomengsun/car-insurance-claim-data}.
The dataset contains various factors that influence the scale and frequency of claims. These factors not only make the analysis process more complicated but also provide a wealth of valuable insights. 
Let the response $Y$ be the claim amount and the covariates $\bX\in\mbR^p$ with $p=39$. These covariates include the original continuous variables and the one-hot vector of categorical variables. 
We consider the following regression model:
\begin{align}
    \label{esreg}
    Q_{\tau}(Y_i\mid\bX_i) = \bX_i^\top \bbeta(\tau), \qquad ES_{\tau}(Y_i\mid\bX_i) = \bX_i^\top \theta(\tau), \qquad i=1,\dots,n=7657.
\end{align}
We aim to investigate the impact of a specific factor on extremely large claims. Thus, we select $\tau=0.9,0.95,0.99$. To facilitate batch-wise modeling, the sample was randomly shuffled and partitioned into 20 batches, with the first 19 batches each containing 382 observations and the 20th batch containing 399 observations.

\begin{figure}[htbp]
    \centering
    \includegraphics[width=0.9\linewidth]{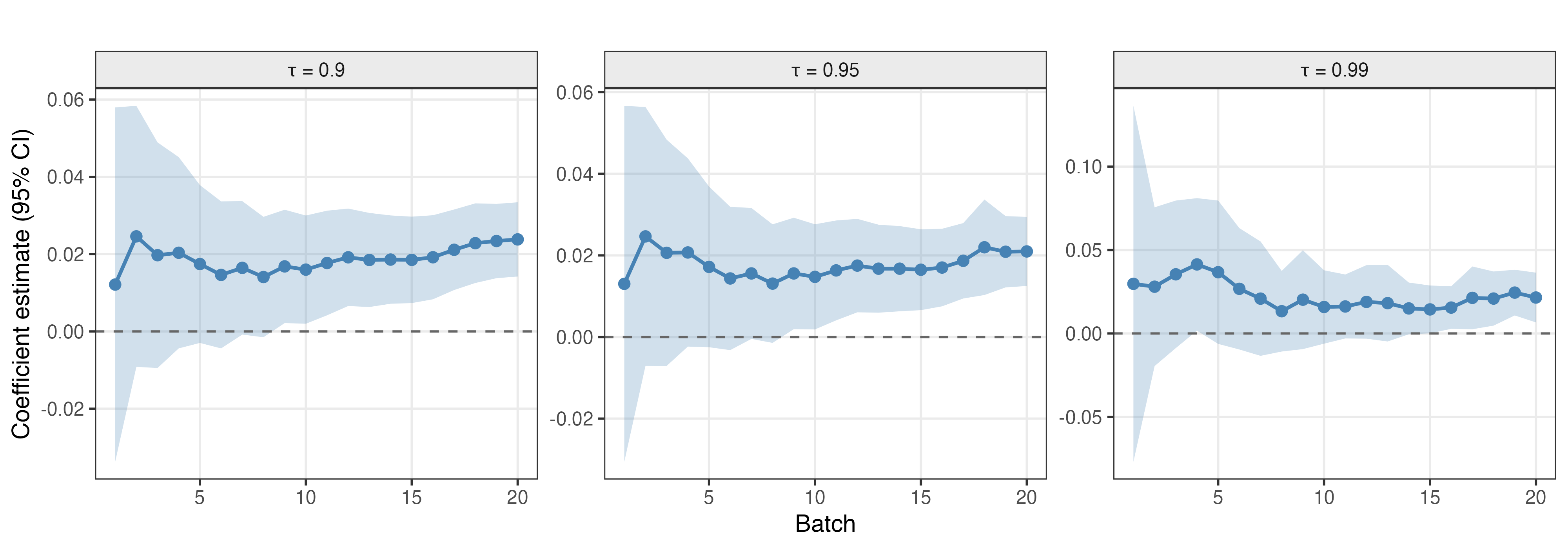}
    \caption{The 95\% confidence intervals for the estimated coefficients of 
    \texttt{MVR\_PTS} across batches, derived from upper-tail expected shortfall regressions 
    with $\tau \in \{0.9, 0.95, 0.99\}$.}
    \label{fig4}
\end{figure}

We then fit the online $\ell_1$-penalized ES regression to examine claim severity among high-risk drivers under extreme claim scenarios, corresponding to the upper tail of the conditional distribution. The primary variable of interest was \texttt{MVR\_PTS}. A confidence interval that excludes zero indicates a statistically significant effect of \texttt{MVR\_PTS} on claim amounts at the corresponding quantile level. As additional data batches were incorporated, the coefficient estimates and their confidence intervals stabilized. The results in Figure~\ref{fig4} show that drivers with higher violation points exhibit significantly different claim amounts in extreme scenarios, further underscoring the role of driving violations in insurance risk pricing.

\section{Conclusion}
In this paper, we propose a renewable estimation and inference framework for the high-dimensional ES regression model, which adapts to sequential updating and storage constraints in the online setting. We construct a surrogate loss function via quadratic approximation to obtain the online estimator. We further adopt an online debiased estimator by using a bias-correction term tailored to the ES orthogonal score structure. Theoretically, we establish the non-asymptotic error bounds for the online estimator and asymptotic normality for the online debiased estimator. Simulation results confirm the method's good finite-sample performance in several metrics, which are close to the offline benchmark with full data. Application of the car insurance claim dataset verifies its ability to identify key tail risk factors and support real-time model updates.

This work enriches the theoretical system of high-dimensional ES regression and provides an efficient tool for real-time tail risk management in big data scenarios, with broad applications in financial regulation, portfolio optimization, insurance pricing, and credit risk measurement. Despite the innovations of our work, there remain questions for future research. Directions include extending our framework to nonlinear or multivariate ES regression, as well as expanding to other tail risk measures to form a unified renewable inference system.

Appendix
\begin{appendices}
\label{app}
\allowdisplaybreaks[3]

\renewcommand{\theequation}{\thesection.\arabic{equation}}

\setcounter{table}{0}
\renewcommand{\thetable}{\thesection\arabic{table}}
\setcounter{figure}{0}
\renewcommand{\thefigure}{\thesection\arabic{figure}}
\setcounter{algorithm}{0}
\renewcommand{\thealgorithm}{\thesection\arabic{algorithm}}

\section{Proof of main theorems}

\begin{proof}[Proof of Theorem \ref{thmesterrbound}]
    Let $\mbC(l)=\{\bdelta\in\mbR^p: \Vert\bdelta\Vert_1 \leq l\Vert\bdelta\Vert_{\bSigma} \}$ be a cone set. Denote the error vector $\bdelta^{(s)}=\hat{\btheta}^{(s)}-\btheta^*$. On one hand, we have
\begin{align}
    \label{mcl2leftside}
    &\langle \nabla\mcL_2(\hat{\btheta}^{(2)})-\nabla\mcL_2(\btheta^*), \hat{\btheta}^{(2)}-\btheta^* \rangle \notag \\
    = & \frac{n_1}{N_2}\tau^2(\hat{\btheta}^{(2)}-\hat{\btheta}^{(1)})^\top \hat{\bSigma}_1 (\hat{\btheta}^{(2)}-\btheta^*) - \frac{n_1}{N_2}\tau^2(\btheta^*-\hat{\btheta}^{(1)})^\top \hat{\bSigma}_1 (\hat{\btheta}^{(2)}-\btheta^*) \notag \\
    & + \frac{n_2}{N_2}\big\{\nabla L(\hat{\bbeta}^{(2)}, \hat{\btheta}^{(2)};\mcD_2) - \nabla L(\hat{\bbeta}^{(2)}, \btheta^*;\mcD_2)\big\}^\top (\hat{\btheta}^{(2)}-\btheta^*) \notag \\
    = & \frac{n_1}{N_2}\tau^2(\hat{\btheta}^{(2)}-\btheta^*)^\top \hat{\bSigma}_1 (\hat{\btheta}^{(2)}-\btheta^*) + \frac{n_2}{N_2}\tau^2(\hat{\btheta}^{(2)}-\btheta^*)^\top \hat{\bSigma}_2 (\hat{\btheta}^{(2)}-\btheta^*) \notag \\
    = & \tau^2 \{\bdelta^{(2)}\}^\top \tilde{\bSigma}_2 \bdelta^{(2)}.
\end{align}

On the other hand, by the KKT condition, there exists a subgradient $\bkappa_2\in \partial\Vert\hat{\btheta}^{(2)}\Vert_1$ satisfying $\nabla\mcL_2(\hat{\btheta}^{(2)}) + \tau\lambda_e^{(2)}\bkappa_2=0$. Thus we have
\begin{align}
    \label{mcl2rightside}
    &\langle \nabla\mcL_2(\hat{\btheta}^{(2)})-\nabla\mcL_2(\btheta^*), \hat{\btheta}^{(2)}-\btheta^* \rangle \notag \\
    =& -\tau\lambda_e^{(2)}\bkappa_2^\top (\hat{\btheta}^{(2)}-\btheta^*) - \frac{n_2}{N_2} \tau^2 (\btheta^*-\hat{\btheta}^{(1)})^\top \hat{\bSigma}_1 (\hat{\btheta}^{(2)}-\btheta^*) \notag \\
    & -\frac{n_2}{N_2}\nabla L(\hat{\bbeta}^{(2)}, \btheta^*; \mcD_2)^\top (\hat{\btheta}^{(2)}-\btheta^*) \notag \\
    :=& A_1 + A_2 + A_3
\end{align}
We then analyze $A_1,\ A_2,\ A_3$ in \eqref{mcl2rightside}. 
By the definition of subgradient, we have
\begin{align}
    \label{a1bound}
    A_1 \leq \tau \lambda_e^{(2)} \big(\Vert\btheta^*\Vert_1 - \Vert\hat{\btheta}^{(2)}\Vert_1\big) \leq \tau \lambda_e^{(2)}  \big(\Vert\bdelta^{(2)}_{S}\Vert_1 - \Vert\bdelta^{(2)}_{S^c}\Vert_1\big).  
\end{align}
For $A_2$, apply the decomposition
\begin{align}
    \label{a2decomp}
    A_2 =  \frac{n_1}{N_2} \tau^2 \bdelta^{(1)\top}\bSigma\bdelta^{(2)} +  \frac{n_1}{N_2} \tau^2 \bdelta^{(1)\top}(\hat{\bSigma}_1-\bSigma)\bdelta^{(2)}  
    := & A_{21} + A_{22}.
\end{align}
Further note that $\tau\Vert\bdelta^{(1)}\Vert_{\bSigma} \leq 4\rho^{-1}s_e^{1/2}\lambda_e^{(1)}$. Given $\lambda_e^{(2)}\sqrt{N_2} \asymp \lambda_e^{(1)}\sqrt{n_1}$, it then follows from \eqref{a2decomp} that
\begin{align}
    \label{a21bound}
    A_{21} \leq 4\tau\rho^{-1}s_e^{1/2}\lambda_e^{(2)}  \Vert\bdelta^{(2)}\Vert_{\bSigma}.  
\end{align}
Further by Lemma \ref{lemmacovest}, we know that $\Vert\hat{\bSigma}_1-\bSigma\Vert_{\max} \leq M_4^{1/2}\sigma_{\bX}^2\sqrt{6\log(p)/n_1} + B_{\bX}^2\log(p)/n_1$ holds with probability no lower than $1-2/p$. Using Lemma \ref{lemmaquadraformbound}, as long as $n_1\geq (B_{\bX}\sigma_{\bX})^4M_4^{-1}\log(p)$, we get
\begin{align}
    \label{a22bound}
    A_{22} \leq &\frac{n_1}{N_2}\tau^2 \Vert\bdelta^{(1)}\Vert_1 \Vert\hat{\bSigma}_1-\bSigma\Vert_{\max} \Vert\bdelta^{(2)}\Vert_1 \notag \\
    \leq & \frac{n_1}{N_2} 52\tau \rho^{-1} M_4^{1/2}\sigma_{\bX}^2s_e\lambda_e^{(1)}\sqrt{\frac{6\log(p)}{n_1}} \Vert\bdelta^{(2)}\Vert_1 
    \asymp \tau \lambda_e^{(2)} s_e\sqrt{\frac{\log(p)}{N_2}} \Vert\bdelta^{(2)}\Vert_1. 
\end{align}
Since $s_e\sqrt{\log(p)/N_2}=o(1)$, given large enough $N_2$, we know that $A_{22}$ can be controlled by $\tau\lambda_e^{(2)} \Vert\bdelta^{(2)}\Vert_1/4$. Combining this observation with \eqref{a2decomp} and \eqref{a21bound} yields that
\begin{align}
    \label{a2bound}
    A_2 \leq \frac{4}{\rho} \tau s_e^{1/2} \lambda_e^{(2)}  \Vert\bdelta^{(2)}\Vert_{\bSigma} + \frac{1}{4} \tau \lambda_e^{(2)} \Vert\bdelta^{(2)}\Vert_1.
\end{align}
Recall that $e_{si}(\bbeta)=Z_{si}(\bbeta)-\tau\bX_{si}^\top\btheta^*$. For $A_3$, note that $\mbE\big\{e_{si}(\bbeta^*)\big\} = \mbE\big[\mbE\{e_{si}(\bbeta^*)\mid\bX_{si}\}\big]=0$. It then follows that 
\begin{align}
    \label{a3decomp}
    -\nabla L(\hat{\bbeta}^{(2)},\btheta^*;\mcD_2) 
    =& \frac{\tau}{n_2} \sum_{i=1}^{n_2} e_{2i}(\bbeta^*) \bX_{2i} + \frac{\tau}{n_2} \sum_{i=1}^{n_2} (1-\mbE)\{e_{2i}(\hat{\bbeta}^{(2)})-e_{2i}(\bbeta^*)\} \bX_{2i} \notag \\ 
    &+ \tau\mbE\{e_{2i}(\hat{\bbeta}^{(2)})\bX_{2i} \} 
\end{align}
Under Assumption \ref{assumpepsilon} and \ref{assumpxboundRME}, by Lemma S 1.1, S 1.2 and S 1.3 in \cite{zhang2023hdesreg}, we have
\begin{align*}
    & \sup_{\bbeta\in\bbeta^* + \mbB_1(r_1)} \bigg\Vert \frac{1}{n_s}\sum_{i=1}^{n_s}(1-\mbE)\{e_{si}(\bbeta)-e_{si}(\bbeta^*)\} \bX_{si} \bigg\Vert_{\infty} \leq 4 B_{\bX} r_1\bigg\{ \sigma_{\bX}\sqrt{\frac{6\log(p)}{n_s}} + B_{\bX}\frac{\log(p)}{n_s}\bigg\}, \\
    & \bigg\Vert\frac{1}{n_s}\sum_{i=1}^{n_s} e_{si}(\bbeta^*)\bX_{si}\bigg\Vert_{\infty} \leq 2\sigma_{\varepsilon}\sigma_{\bX}\sqrt{\frac{\log(p)}{n_s}} + 4 q_{\varepsilon}B_{\bX}\frac{\log(p)}{n_s}
\end{align*}
holds with probability at least $1-4p^{-1}$ and for any $r_0>0$
\begin{align*}
    \sup_{\bbeta\in\bbeta^* + \mbB_{\bSigma}(r_0)} \Vert\mbE\{e_{si}(\bbeta)\bSigma^{-1/2}\bX_{si}\}\Vert_2 \leq \frac{1}{2}F_u M_3 r_0^2.
\end{align*}

Let $\eta_s=(n_s/N_s)^{1/2}$. Given that $\lambda_e^{(2)}\geq 4\sigma_{\bX}(2q_{\epsilon}+5B_{\bX}r_1)\{\log(p)/N_2\}^{1/2}$ and $N_2 \geq 2(q_{\epsilon}/\sigma_{\epsilon})^2(B_{\bX}/\sigma_{\bX})^2\log(p)$, it follows from \eqref{a3decomp} that 
\begin{align}
    \label{a3bound}
    &-\nabla L(\hat{\bbeta}^{(2)},\btheta^*;\mcD_2)^\top (\hat{\btheta}^{(2)}-\btheta) \notag \\
    \leq & \tau \left\{ \bigg\Vert\frac{1}{n_2}\sum_{i=1}^{n_2} e_{2i}(\bbeta^*)\bX_{2i}\bigg\Vert_{\infty} + \sup_{\bbeta\in\bbeta^* + \mbB_1(r_1)} \bigg\Vert \frac{1}{n_2}\sum_{i=1}^{n_2}(1-\mbE)\{e_{2i}(\bbeta)-e_{2i}(\bbeta^*)\} \bX_{2i} \bigg\Vert_{\infty} \right\} \Vert\bdelta^{(2)}\Vert_1 \notag \\
    & + \tau \sup_{\bbeta\in\bbeta^* + \mbB_{\bSigma}(r_0)}\Vert\mbE\{e_{si}(\bbeta)\bSigma^{-1/2}\bX_{si}\}\Vert_2  \Vert\bdelta^{(2)}\Vert_{\bSigma} \notag \\
    \leq & \frac{1}{4}\tau \eta_2\lambda_e^{(2)} \Vert\bdelta^{(2)}\Vert_1 + \frac{1}{2}\tau F_u M_3 r_0^2 \Vert\bdelta^{(2)}\Vert_{\bSigma}
\end{align}
holds with probability at least $1-4p^{-1}$. Following the decomposition \eqref{mcl2rightside} and inequalities \eqref{a1bound}, \eqref{a2bound} and \eqref{a3bound}, we derive the upper bound for $\langle \nabla\mcL_2(\hat{\btheta}^{(2)})-\nabla\mcL_2(\btheta^*), \bdelta^{(2)} \rangle$. Combining the upper bound with \eqref{mcl2leftside} yields that 
\begin{align}
    \label{mcl2upperbound}
    2\tau^2\bdelta^{(2)\top} \tilde{\bSigma}_2 \bdelta^{(2)} \leq \tau\lambda_e^{(2)} \big( 3\Vert\bdelta^{(2)}_{S}\Vert_1-\Vert\bdelta^{(2)}_{S^c}\Vert_1 \big) + 8\tau \rho^{-1} s_e^{1/2} \lambda_e^{(2)}\Vert\bdelta^{(2)}\Vert_{\bSigma} + \tau F_uM_3 r_0^2 \Vert\bdelta^{(2)}\Vert_{\bSigma}, 
\end{align}
which further indicates the cone-type restriction $\Vert\bdelta^{(2)}_{S^c}\Vert_1 \leq 3\Vert\bdelta^{(2)}_{S}\Vert_1 + \big(8\rho^{-1}s_e^{1/2} + F_uM_3r_0^2/\lambda_e^{(2)}\big) \Vert\bdelta^{(2)}\Vert_{\bSigma}$. Provided that $\lambda_q^{(2)}\geq F_uM_3\rho s_0^{-1/2}r_0^2$, we have
\begin{align}
    \label{delta2conerestriction}
    \Vert\bdelta^{(2)}\Vert_1 \leq& 4\Vert\bdelta^{(2)}_{S}\Vert_1 + 9\rho^{-1}s_e^{1/2}\Vert\bdelta^{(2)}\Vert_{\bSigma} \notag \\ 
    \leq& 4s_e^{1/2}\Vert\bdelta^{(2)}\Vert_2 + 9\rho^{-1}s_e^{1/2}\Vert\bdelta^{(2)}\Vert_{\bSigma} \leq 13\rho^{-1} s_e^{1/2}\Vert\bdelta^{(2)}\Vert_{\bSigma},
\end{align}
where the second inequality holds by H{\"o}lder's inequality and the third inequality holds by Assumption \ref{assumpxboundRME}. The inequality \eqref{delta2conerestriction} implies that $\bdelta^{(2)}$ falls in the cone area $\mbC(l_1)$ with $l_1=13\rho^{-1}s_e^{1/2}$. By Lemma A.4 in \cite{zhang2023hdesreg},  we have
\begin{align}
    \label{mcl2lowerbound}
    2\bdelta^{(2)\top} \tilde{\bSigma}_2 \bdelta^{(2)} \geq \Vert\bdelta^{(2)}\Vert_{\bSigma}^2 
\end{align}
holds with probability at least $1-(2p)^{-1}$, provided that $N_2 \geq \max\{CM_4\rho^{-2}\sigma_{\bX}^2(s_e+1)\log(2p), B_{\bX}^2\sigma_{\bX}^{-2}\log(2p)\}$ for some large enough constant $C$. We then conclude from \eqref{mcl2upperbound}, \eqref{delta2conerestriction} and \eqref{mcl2lowerbound} that 
\begin{align*}
    \tau^2 \Vert\bdelta^{(2)}\Vert_{\bSigma}^2/2 \leq 3\tau\lambda_e^{(2)}\Vert\bdelta^{(2)}_{S}\Vert_1/2 + \frac{1}{2}\tau F_uM_3r_0^2\Vert\bdelta^{(2)}\Vert_{\bSigma} \leq \Big(\frac{3}{2\rho}s_e^{1/2}\lambda_e^{(2)} + \frac{1}{2\rho}s_e^{1/2}\lambda_e^{(2)} \Big) \tau \Vert\bdelta^{(2)}\Vert_{\bSigma}.
\end{align*}
By deleting a factor of $\tau \Vert\bdelta^{(2)}\Vert_{\bSigma}$ from both sides, we then derive
\begin{align}
    \label{delta2signorm}
    \tau\Vert\hat{\btheta}^{(2)}-\btheta^*\Vert_{\bSigma} \leq \frac{4}{\rho}s_e^{1/2}\lambda_e^{(2)}.
\end{align}
Combining the cone-type restriction \eqref{delta2conerestriction} with the bound \eqref{delta2signorm} implies 
\begin{align}
    \label{delta2l1norm}
     \tau\Vert\hat{\btheta}^{(2)}-\btheta^*\Vert_{1} \leq \frac{52}{\rho^2}s_e\lambda_e^{(2)}.
\end{align}

Analogously, suppose the error bounds \eqref{delta2signorm} and \eqref{delta2l1norm} hold for $s=1,2,\dots,b-1$ with probability approaching $1$. We then demonstrate the error bounds at time $b$, that is, 
\begin{align}
    \label{deltabsignorm}
    \tau\Vert\hat{\btheta}^{(b)}-\btheta^*\Vert_{\bSigma} \leq \frac{4}{\rho}s_e^{1/2}\lambda_e^{(b)}, \\
     \label{deltabl1norm}
    \tau\Vert\hat{\btheta}^{(b)}-\btheta^*\Vert_1 \leq \frac{52}{\rho^2}s_e\lambda_e^{(b)}. 
\end{align}
The analysis procedure here is similar to that under $s=2$, thus we briefly demonstrate the proof next, with some adjustments in notations from $s=2$ to $s=b$. 
Similarly, we have the decomposition
\begin{align}
    \label{mclbleftside}
    &\langle \nabla\mcL_b(\hat{\btheta}^{(b)})-\nabla\mcL_b(\btheta^*), \hat{\btheta}^{(b)}-\btheta^* \rangle \notag \\
    = & \frac{N_{b-1}}{N_b}\tau^2(\hat{\btheta}^{(b)}-\hat{\btheta}^{(b-1)})^\top \tilde{\bSigma}_{b-1} (\hat{\btheta}^{(b)}-\btheta^*) -  \frac{N_{b-1}}{N_b}\tau^2(\btheta^*-\hat{\btheta}^{(b-1)})^\top \tilde{\bSigma}_{b-1} (\hat{\btheta}^{(b)}-\btheta^*) \notag \\
    & + \frac{n_b}{N_b}\big\{\nabla L(\hat{\bbeta}^{(b)}, \hat{\btheta}^{(b)};\mcD_b) - \nabla L(\hat{\bbeta}^{(b)}, \btheta^*;\mcD_b)\big\}^\top (\hat{\btheta}^{(b)}-\btheta^*) \notag \\
    = & \frac{N_{b-1}}{N_b}\tau^2(\hat{\btheta}^{(b)}-\btheta^*)^\top \tilde{\bSigma}_{b-1} (\hat{\btheta}^{(b)}-\btheta^*) + \frac{n_b}{N_b}\tau^2(\hat{\btheta}^{(b)}-\btheta^*)^\top \hat{\bSigma}_b (\hat{\btheta}^{(b)}-\btheta^*) \notag \\
    = & \tau^2 \{\bdelta^{(b)}\}^\top \tilde{\bSigma}_b \bdelta^{(b)},
\end{align}
and 
\begin{align}
    \label{mclbrightside}
    &\langle \nabla\mcL_b(\hat{\btheta}^{(b)})-\nabla\mcL_b(\btheta^*), \hat{\btheta}^{(b)}-\btheta^* \rangle \notag \\
    =& -\tau\lambda_e^{(b)}\bkappa_b^\top (\hat{\btheta}^{(b)}-\btheta^*) - \frac{N_{b-1}}{N_b} \tau^2 (\btheta^*-\hat{\btheta}^{(b-1)})^\top \tilde{\bSigma}_{b-1} (\hat{\btheta}^{(b)}-\btheta^*) \notag \\
    & -\frac{n_b}{N_b}\nabla L(\hat{\bbeta}^{(b)}, \btheta^*; \mcD_b)^\top (\hat{\btheta}^{(b)}-\btheta^*) \notag \\
    :=& D_1 + D_2 + D_3,
\end{align}
where $\bkappa_b\in \partial \Vert\hat{\btheta}^{(b)}\Vert_1$ is some subgradient. For $D_1$, by triangle inequality we have
\begin{align}
    \label{D1bound}
    D_1 \leq \tau \lambda_e^{(b)} \big(\Vert\btheta^*\Vert_1 - \Vert\hat{\btheta}^{(b)}\Vert_1\big) \leq \tau \lambda_e^{(b)}  \big(\Vert\bdelta^{(b)}_{S}\Vert_1 - \Vert\bdelta^{(b)}_{S^c}\Vert_1\big).   
\end{align}
For $D_2$, by triangle inequality and H{\"o}lder's inequality, we get
\begin{align}
    \label{D2decomp}
    D_2 =& \frac{N_{b-1}}{N_b} \tau^2 \bdelta^{(b-1)\top} \tilde{\bSigma}_{b-1} \bdelta^{(b)} \notag \\
    =& \frac{N_{b-1}}{N_b} \tau^2 \bdelta^{(b-1)\top} \bSigma \bdelta^{(b)} + \frac{N_{b-1}}{N_b} \tau^2 \bdelta^{(b-1)\top} \big( \tilde{\bSigma}_{b-1}-\bSigma \big) \bdelta^{(b)} \notag \\
    \leq& \frac{N_{b-1}}{N_b} \tau^2 \Vert\bSigma^{1/2}\bdelta^{(b-1)}\Vert_2 \Vert\bSigma^{1/2}\bdelta^{(b)}\Vert_2 \notag \\
    &+ \frac{N_{b-1}}{N_b} \tau^2 \Vert\big(\tilde{\bSigma}_{b-1}-\bSigma\big)^{1/2}\bdelta^{(b-1)}\Vert_2 \Vert\big(\tilde{\bSigma}_{b-1}-\bSigma\big)^{1/2}\bdelta^{(b)}\Vert_2  \notag \\
    :=& D_{21} + D_{22}.
\end{align}
Note that $\tau\Vert\bdelta^{(b-1)}\Vert_{\bSigma} \leq 4\rho^{-1}s_e^{1/2}\lambda_e^{(b-1)}$, given $\lambda_e^{(b)}\sqrt{N_b} \asymp \lambda_e^{(b-1)}\sqrt{N_{b-1}}$, we have
\begin{align}
    \label{D21bound}
    D_{21} = \frac{N_{b-1}}{N_b} \tau^2 \Vert\bdelta^{(b-1)}\Vert_{\bSigma} \Vert\bdelta^{(b)}\Vert_{\bSigma} \leq 4\tau\rho^{-1}s_e^{1/2}\lambda_e^{(b)} \Vert\bdelta^{(b)}\Vert_{\bSigma}.    
\end{align}
Moreover, by Lemma \ref{lemmacovest}, we get $\Vert\tilde{\bSigma}_{b-1}-\bSigma\Vert_{\max} \leq M_4^{1/2}\sigma_{\bX}^2\sqrt{6\log(p)/N_{b-1}} + B_{\bX}^2\log(p)/N_{b-1}$ holds with probability no lower than $1-2/p$. Using Lemma \ref{lemmaquadraformbound}, as long as $N_{b-1}\geq (B_{\bX}\sigma_{\bX})^4M_4^{-1}\log(p)$,  it follows that
\begin{align}
    \label{D22bound}
    D_{22} \leq &\frac{N_{b-1}}{N_b}\tau^2 \Vert\bdelta^{(b-1)}\Vert_1 \Vert\tilde{\bSigma}_{b-1}-\bSigma\Vert_{\max} \Vert\bdelta^{(b)}\Vert_1 \notag \\
    \leq & \frac{N_{b-1}}{N_b} 52\tau \rho^{-1} M_4^{1/2}\sigma_{\bX}^2s_e\lambda_e^{(b-1)}\sqrt{\frac{6\log(p)}{N_{b-1}}} \Vert\bdelta^{(b)}\Vert_1 
    \asymp \tau \lambda_e^{(b)} s_e\sqrt{\frac{\log(p)}{N_b}} \Vert\bdelta^{(b)}\Vert_1. 
\end{align}
For a large enough $N_b$, $D_{22}$ is well-controlled by $\tau\lambda_e^{(b)}\Vert\bdelta^{(b)}\Vert_1/4$. 
\begin{align}
    \label{D2bound}
    D_2 \leq 4\tau\rho^{-1}s_e^{1/2}\lambda_e^{(b)} \Vert\bdelta^{(b)}\Vert_{\bSigma} + 4^{-1}\tau\lambda_e^{(b)}\Vert\bdelta^{(b)}\Vert_1.
\end{align}
For $D_3$, 
\begin{align}
    \label{D3decomp}
    -\nabla L(\hat{\bbeta}^{(b)},\btheta^*;\mcD_b) 
    =& \frac{\tau}{n_b} \sum_{i=1}^{n_b} e_{bi}(\bbeta^*) \bX_{bi} + \frac{\tau}{n_b} \sum_{i=1}^{n_b} (1-\mbE)\{e_{bi}(\hat{\bbeta}^{(b)})-e_{bi}(\bbeta^*)\} \bX_{bi} \notag \\ 
    &+ \tau\mbE\{e_{bi}(\hat{\bbeta}^{(b)})\bX_{bi} \}  
\end{align}
Similarly, using Lemma S 1.1, S 1.2 and S 1.3 in \cite{zhang2023hdesreg},
given $\hat{\bbeta}^{(b)} \in \bbeta^* + \mbB_{\bSigma}(r_0) \cap
 \mbB_1(r_1)$, $\lambda_e^{(b)}\geq 4\sigma_{\bX}(2q_{\epsilon}+5B_{\bX}r_1)\{\log(p)/N_b\}^{1/2}$ and $N_b \geq 2(q_{\epsilon}/\sigma_{\epsilon})^2(B_{\bX}/\sigma_{\bX})^2\log(p)$,
\begin{align}
    \label{D3bound}
    &-\nabla L(\hat{\bbeta}^{(b)},\btheta^*;\mcD_b)^\top (\hat{\btheta}^{(b)}-\btheta) \notag \\
    \leq & \tau \left\{ \bigg\Vert\frac{1}{n_b}\sum_{i=1}^{n_b} e_{bi}(\bbeta^*)\bX_{bi}\bigg\Vert_{\infty} + \sup_{\bbeta\in\bbeta^* + \mbB_1(r_1)} \bigg\Vert \frac{1}{n_b}\sum_{i=1}^{n_b}(1-\mbE)\{e_{bi}(\bbeta)-e_{bi}(\bbeta^*)\} \bX_{bi} \bigg\Vert_{\infty} \right\} \Vert\bdelta^{(b)}\Vert_1 \notag \\
    & + \tau \sup_{\bbeta\in\bbeta^* + \mbB_{\bSigma}(r_0)}\Vert\mbE\{e_{bi}(\bbeta)\bSigma^{-1/2}\bX_{bi}\}\Vert_2  \Vert\bdelta^{(b)}\Vert_{\bSigma} \notag \\
    \leq & \frac{1}{4}\tau \eta_b\lambda_e^{(b)} \Vert\bdelta^{(b)}\Vert_1 + \frac{1}{2}\tau F_u M_3 r_0^2 \Vert\bdelta^{(b)}\Vert_{\bSigma}
\end{align}
holds with probability at least $1-4p^{-1}$. Combining \eqref{D1bound}, \eqref{D2bound}, \eqref{D3bound} with \eqref{mclbleftside} and \eqref{mclbrightside}, we conclude that 
\begin{align}
    \label{mclbupperbound}
    2\tau^2\bdelta^{(b)\top} \tilde{\bSigma}_b \bdelta^{(b)} \leq \tau\lambda_e^{(b)} \big( 3\Vert\bdelta^{(b)}_{S}\Vert_1-\Vert\bdelta^{(b)}_{S^c}\Vert_1 \big) + 8\tau \rho^{-1} s_e^{1/2} \lambda_e^{(b)}\Vert\bdelta^{(b)}\Vert_{\bSigma} + \tau F_uM_3 r_0^2 \Vert\bdelta^{(b)}\Vert_{\bSigma}.
\end{align}
Since the left side of \eqref{mclbupperbound} is non-negative, by Assumption \ref{assumpxboundRME} and  H{\"o}lder's inequality, we have
\begin{align}
    \label{deltabconerestr}
    \Vert\bdelta^{(b)}\Vert_1 \leq 4 s_e^{1/2}\Vert\bdelta^{(b)}\Vert_2 + 9\rho^{-1}s_e^{1/2}\Vert\bdelta^{(b)}\Vert_{\bSigma} \leq 13 \rho^{-1} s_e^{1/2} \Vert\bdelta^{(b)}\Vert_{\bSigma}.
\end{align}
provided that $\lambda_e^{(s)} \geq \rho F_uM_3s_e^{-1/2}r_0^2$.
Therefore, $\bdelta^{(b)}$ lies in the cone set $\mbC(l)$ with $l=13\rho^{-1}s_e^{1/2}$. Again, by Lemma A.4 in \cite{zhang2023hdesreg}, if $N_b \geq \max\{CM_4\rho^{-2}\sigma_{\bX}^2(s_e+1)\log(2p), B_{\bX}^2\sigma_{\bX}^{-2}\log(2p)\}$ for some large enough constant $C$, we have
\begin{align}
    \label{mclblowerbound}
    2\bdelta^{(b)\top} \tilde{\bSigma}_b \bdelta^{(b)} \geq \Vert\bdelta^{(b)}\Vert_{\bSigma}^2 
\end{align}
holds with probability no lower than $1-(2p)^{-1}$. We then conclude from \eqref{mclbupperbound}, \eqref{deltabconerestr} and \eqref{mclblowerbound} that 
\begin{align*}
    \tau^2 \Vert\bdelta^{(b)}\Vert_{\bSigma}^2/2 \leq 3\tau\lambda_e^{(b)}\Vert\bdelta^{(b)}_{S}\Vert_1/2 + \frac{1}{2}\tau F_uM_3r_0^2\Vert\bdelta^{(b)}\Vert_{\bSigma} \leq \Big(\frac{3}{2\rho}s_e^{1/2}\lambda_e^{(b)} + \frac{1}{2\rho}s_e^{1/2}\lambda_e^{(b)} \Big) \tau \Vert\bdelta^{(b)}\Vert_{\bSigma}.
\end{align*}
By deleting a factor of $\tau \Vert\bdelta^{(2)}\Vert_{\bSigma}$ from both sides, we then demonstrate the $\bSigma$-norm error bounds \eqref{deltabsignorm} for $\hat{\btheta}^{(b)}$. Moreover, through the cone-type restriction \eqref{deltabconerestr}, the $\ell_1$-norm error bounds \eqref{deltabl1norm} is confirmed.

\end{proof}

\begin{proof}[Proof of Corollary \ref{crlryesterr}] 
   By Proposition \ref{propbetaesterr}, we know that $\hat{\bbeta}^{(s)} \in \bbeta^* + \mbB_{\bSigma}(r_0) \cap \mbB_1(r_1)$ with $r_0=C_{0,s}\{s_0\log(p)/N_s\}^{1/2}$, $r_1=C_{1,s}s_0\{\log(p)/N_s\}^{1/2}$ for some constants $C_{1,2}$ and $C_{2,2}$ at a large probability $1-4/p$. Substituting $r_0,r_1$ in Theorem \ref{thmesterrbound}, with $N_s\geq C'^2s_0^2\log(p)$ for some constant $C'$, the penalization parameter $\lambda_e^{(s)}$ then satisfies 
   $\lambda_e^{(s)} \geq C_3\sqrt{\log(p)/N_s}$ with $C_3=\max\{4\sigma_{\bX}(2q_{\epsilon}+5B_{\bX}C_{1,s}/C),F_uM_3\rho C_{0,s}^2/C\}$. Further note that $\lambda_e^{(s)}\leq C_4\sqrt{\log(p)/N_s}$ for some large enough constant $C_4$, it then follows that
   \[
   \tau\Vert\hat{\btheta}^{(s)}-\btheta^*\Vert_{\bSigma} \leq 4C_4\rho^{-1}s_e^{1/2}\sqrt{\log(p)/N_s}, \quad
    \tau\Vert\hat{\btheta}^{(b)}-\btheta^*\Vert_1 \leq 52C_4\rho^{-2}s_e\sqrt{\log(p)/N_s}.
   \]
   with probability at least $1-5/p$. Thus the statement holds.
   
\end{proof}

\begin{proof}[Proof of Theorem \ref{theoremasympnorm}]
For any $\theta_j\in\mbR$, note that 
\begin{align*}
    \hat{S}_{n_s}^{(s)}(\theta_j) - \hat{S}_{n_s}^{(s)}(\theta_j^*) = -\frac{\tau}{n_s}\sum_{i=1}^{n_s}X_{si,j}\big(X_{si,j}-\bX_{si,-j}^\top\hat{\bgamma}^{(s)}\big)(\theta_j-\theta_j^*).
\end{align*}
Recall that $\hat{H}_{n_s}^{(s)}=-\tau/n_s\sum_{i=1}^{n_s}X_{si,j}\big(X_{si,j}-\bX_{si,-j}^\top\hat{\bgamma}^{(s)}\big)$, thus we know that 
\begin{align*}
    \hat{S}_{n_s}^{(s)}(\theta_j) - \hat{S}_{n_s}^{(s)}(\theta_j^*) = \hat{H}_{n_s}^{(s)}(\theta_j-\theta_j^*).
\end{align*}
It then follows that 
\begin{align}
    \label{debiasesterrdecomp}
    \tilde{\theta}_{j}^{(b)}-\theta_j^* =& \hat{\theta}_{j}^{(b)} - \theta_j^* - \tilde{S}_{N_b}/\tilde{H}_{N_b} - N_b^{-1}\sum_{s=1}^b n_s\hat{H}_{n_s}^{(s)}\big( \hat{\theta}_j^{(s)}-\hat{\theta}_j^{(b)} \big) / \tilde{H}_{N_b} \notag \\
    =& \hat{\theta}_{j}^{(b)}-\theta_j^* - \Big\{\tilde{S}_{N_b}-\frac{1}{N_b}\sum_{s=1}^{b}n_s\hat{S}_{n_s}^{(s)}(\theta_j^*)\Big\}/\tilde{H}_{N_b} - \frac{1}{N_b}\sum_{s=1}^{b} n_s \hat{S}_{n_s}^{(s)}(\theta_j^*)/{\tilde{H}_{N_b}}  \notag \\
    &-\sum_{s=1}^b n_s \hat{H}_{n_s}^{(s)}\big(\hat{\theta}_j^{(s)}-\hat{\theta}_j^{(b)}\big)/\tilde{H}_{N_b} \notag \\ 
        =& \hat{\theta}_{j}^{(b)}-\theta_j^* - \frac{N_b^{-1} \sum_{s=1}^bn_s\big\{\hat{S}_{n_s}^{(s)}(\hat{\theta}_j^{(s)})-\hat{S}_{n_s}^{(s)}(\theta_j^*)\big\}}{\tilde{H}_{N_b}} - \frac{N_b^{-1}\sum_{s=1}^bn_s\hat{H}_{n_s}^{(s)}\big(\hat{\theta}_j^{(b)}-\theta_j^{(s)}\big)}{\tilde{H}_{N_b}} \notag \\
    & - \frac{1}{N_b}\sum_{s=1}^{b}n_s\hat{S}_{n_s}^{(s)}(\theta_j^*)/\tilde{H}_{N_b} \notag \\
    =& \Bigg\{\hat{\theta}_{j}^{(b)}-\theta_j^* - \frac{N_b^{-1} \sum_{s=1}^bn_s\hat{H}_{n_s}^{(s)}\big(\hat{\theta}_j^{(s)}-\theta_j^*\big)+ N_b^{-1}\sum_{s=1}^bn_s\hat{H}_{n_s}^{(s)}\big(\hat{\theta}_j^{(b)}-\theta_j^{(s)}\big)}{N_b^{-1}\sum_{s=1}^bn_s\hat{H}_{n_s}^{(s)}}  \Bigg\} \notag \\
    & - \frac{1}{N_b}\sum_{s=1}^{b}n_s\hat{S}_{n_s}^{(s)}(\theta_j^*)/\tilde{H}_{N_b} \notag \\
    =& -\frac{1}{N_b} \sum_{s=1}^{b}n_sS_{n_s}^*/\tilde{H}_{N_b} - \frac{1}{N_b}\sum_{s=1}^{b}n_s\big\{\hat{S}_{n_s}^{(s)}(\theta_j^*)-S_{n_s}^*\big\}/\tilde{H}_{N_b} := E_0 + E_1.  
\end{align}
We first investigate the leading term $E_0$ in \eqref{debiasesterrdecomp}. Denote the "oracle" score as $S_{n_s}^*=S_{n_s}^{(s)}(\theta_j^*,\btheta_{-j}^*,\bbeta^*,\bgamma^*)$. Applying the central limit theorem, we see that 
\begin{align*}
    n_s^{1/2}S_{n_s}^* = n_s^{-1/2}\sum_{i=1}^{n_s}e_{si}(\bbeta^*)\omega_{si} \stackrel{d}{\rightarrow} N\big(0,\sigma_{I}^2\big).
\end{align*}
Recall that $\omega_{si}=X_{si,j}-\bX_{si,-j}^\top{\bgamma^*}$ and 
\begin{align}
    \label{E0vardecomp}
    \tilde{H}_{N_b}=-&\frac{\tau}{N_b}\sum_{s=1}^{b}\sum_{i=1}^{n_s}X_{si,j}\big(X_{si,j}-\bX_{si,-j}^\top\hat{\bgamma}^{(s)}\big) \notag \\
    =&-\frac{\tau}{N_b}\sum_{s=1}^{b}\sum_{i=1}^{n_s}X_{si,j}\omega_{si} + \frac{\tau}{N_b} \sum_{s=1}^{b}\sum_{i=1}^{n_s}X_{si,j}\bX_{si,-j}^\top(\hat{\bgamma}^{(s)}-\bgamma^*).
\end{align}
Firstly note that 
\begin{align*}
    \mbE(X_{si,j}\omega_{si})=\mbE\{(X_{si,-j}^\top\bgamma^*+\omega_{si})\omega_{si}\}=\mbE(\omega_{si}^2)=\sigma_{\omega}^2.
\end{align*}
Therefore, by the law of large numbers, 
\begin{align}
\label{E0varconverge}
N_b^{-1}\sum_{s=1}^{b}\sum_{i=1}^{n_s}X_{si,j}\omega_{si} \stackrel{p}{\rightarrow} \sigma_{\omega}^2.
\end{align}
Moreover, we have
\begin{align}
    \label{E0varerr}
    &\bigg|\frac{1}{n_s}\sum_{i=1}^{n_s}X_{si,j}\bX_{si,-j}^\top(\hat{\bgamma}^{(s)}-\bgamma^*)\bigg| \notag \\
    =& \bigg|\frac{1}{n_s}\sum_{i=1}^{n_s}(\bgamma^{*\top}\bX_{si,-j}\bX_{si,-j}^\top+\omega_{si}\bX_{si,-j}^\top)(\hat{\bgamma}^{(s)}-\bgamma^*)\bigg| \notag \\
    =& \bigg|\sum_{i=1}^{n_s}\omega_{si}\bX_{si,-j}^\top(\hat{\bgamma}^{(s)}-\bgamma^*) + \frac{1}{n_s}\bgamma^{*\top}\mbE\Big\{\frac{1}{n_s}\sum_{i=1}^{n_s}\bX_{si,-j}\bX_{si,-j}^\top\Big\}(\hat{\bgamma}^{(s)}-\bgamma^*) \notag \\
    &+\bgamma^{*\top}(1-\mbE)\Big\{\frac{1}{n_s}\sum_{i=1}^{n_s}\bX_{si,-j}\bX_{si,-j}^\top\Big\}(\hat{\bgamma}^{(s)}-\bgamma^*) \bigg| \notag \\
    \leq&  \Vert\hat{\bSigma}_{-j}^{(s)}-\bSigma_{-j}\Vert_{\max}\Vert\bgamma^*\Vert_1\Vert\hat{\bgamma}^{(s)}-\bgamma^*\Vert_1 + \Vert\bgamma^*\Vert_{\bSigma_{-j}}\Vert\hat{\bgamma}^{(s)}-\bgamma^*\Vert_{\bSigma_{-j}} + \bigg\Vert\frac{1}{n_s}\sum_{i=1}^{n_s}\omega_{si}\bX_{si,-j}\bigg\Vert_{\infty}\Vert\hat{\bgamma}^{(s)}-\bgamma^*\Vert_1.
\end{align}
Recall that we already showed that
\begin{align*}
    &\Vert\hat{\bSigma}_{-j}^{(s)}-\bSigma_{-j}\Vert_{\max} \lesssim \big\{\log(p)/n_s\big\}^{1/2}.
\end{align*} 
By Proposition \ref{propgammaesterr}, we know that
\begin{align*}
    \Vert\hat{\bgamma}^{(s)}-\bgamma^*\Vert_{\bSigma_{-j}} \lesssim \big\{s_r\log(p)/N_s\big\}^{1/2}, \quad \Vert\hat{\bgamma}^{(s)}-\bgamma^*\Vert_1 \lesssim s_r\big\{\log(p)/N_s\big\}^{1/2}.
\end{align*}
Moreover, taking $t=2\log(p-1)$ in Lemma S 1.3 of \cite{zhang2023hdesreg}, we have
\begin{align*}
    \bigg\Vert\frac{1}{n_s}\sum_{i=1}^{n_s}\omega_{si}\bX_{si,-j}\bigg\Vert_{\infty} \lesssim \big\{\log(p)/n_s\big\}^{1/2}.
\end{align*}
In addition, note that 
\begin{align*}
    \Vert\bgamma^*\Vert_{\bSigma_{-j}}^2 =& \bgamma^{*\top}\mbE(\bX_{si,-j}\bX_{si,-j}^\top)\bgamma^* = \bgamma^{*\top}\mbE\big\{\bX_{si,-j}(X_{si,j}-\omega_{si})\big\} \\
    =& \bgamma^{*\top}\mbE\big(\bX_{si,-j}X_{si,j}\big) = \mbE(X_{si,j}^2)-\sigma_{\omega}^2=O(1).
\end{align*}
Moreover, by Cauchy-Schwarz inequality and Assumption \ref{assumpxboundRME} we see that 
\begin{align*}
    \Vert\bgamma^*\Vert_1 \leq s_m^{1/2}\Vert\bgamma^*\Vert_2 \leq \Lambda_p^{-1}s_m^{1/2}\Vert\bgamma^*\Vert_{\bSigma_{-j}}.
\end{align*}
Putting these bounds into \eqref{E0varerr} yields 
\begin{align*}
    \bigg|\frac{1}{n_s}\sum_{i=1}^{n_s}X_{si,j}\bX_{si,-j}^\top(\hat{\bgamma}^{(s)}-\bgamma^*)\bigg| \lesssim \sqrt{\frac{s_r\log(p)}{N_s}} + \sqrt{\frac{s_ms_r^2\log^2(p)}{n_sN_s}}.
\end{align*}
Summing over $s=1,\dots,b$ and using Lemma \ref{lemmabatchsize} yields that
\begin{align}
    \label{E0varerrbound}
     &\bigg|\frac{1}{N_b} \sum_{s=1}^{b}\sum_{i=1}^{n_s}X_{si,j}\bX_{si,-j}^\top(\hat{\bgamma}^{(s)}-\bgamma^*)\bigg| \notag \\ \leq &\frac{1}{N_b}\sum_{s=1}^{b}n_s\bigg|\frac{1}{n_s}\sum_{i=1}^{n_s}X_{si,j}\bX_{si,-j}^\top(\hat{\bgamma}^{(s)}-\bgamma^*)\bigg| \notag \\
     \leq & \frac{1}{N_b}\big\{s_r\log(p)\big\}^{1/2}\left[\sum_{s=1}^{b}\frac{n_s}{\sqrt{N_s}}+\big\{s_ms_r\log(p)\big\}^{1/2}\sum_{s=1}^{b}\sqrt{\frac{n_s}{N_s}}\right] \notag \\
     \leq & \big\{s_r\log(p)/N_b
     \big\}^{1/2} + s_rs_m^{1/2}\log(p)\big[b\{1+\log(N_b/n_1)\}\big]^{1/2}/N_b.
\end{align}
Combing \eqref{E0vardecomp}, \eqref{E0varconverge} and \eqref{E0varerrbound}, as long as $s_r^2s_m\log^2(p)b\{1+\log(N_b/n_1)\}=o(N_b)$, we conclude that $\tilde{H}_{N_b}\stackrel{p}{\rightarrow}-\tau\sigma_{\omega}^2$. Finally, by Slutsky's theorem, 
\begin{align*}
    E_0= \frac{-N_b^{-1/2}\sum_{s=1}^{b}n_s S_{n_s}^*}{\tilde{H}_{N_b}}\stackrel{d}{\rightarrow}N(0,\sigma_{I}^2/\sigma_{\omega}^4).
\end{align*}

Turn to the remainder $E_1$ in \eqref{debiasesterrdecomp}. By the decomposition
\begin{align}
    \label{scorespproxdecomp}
    &S_{n_s}^{(s)}(\theta_j^*,\hat{\btheta}_{-j}^{(s)},\hat{\bbeta}^{(s)},\hat{\bgamma}^{(s)}) - S_{n_s}^{(s)}(\theta_j^*,\btheta_{-j}^*,\bbeta^*,\bgamma^*) \notag \\
    =& S_{n_s}^{(s)}(\theta_j^*,\hat{\btheta}_{-j}^{(s)},\hat{\bbeta}^{(s)},\hat{\bgamma}^{(s)}) - S_{n_s}^{(s)}(\theta_j^*,\hat{\btheta}_{-j}^{(s)},\hat{\bbeta}^{(s)},\bgamma^*) + S_{n_s}^{(s)}(\theta_j^*,\hat{\btheta}_{-j}^{(s)},\hat{\bbeta}^{(s)},\bgamma^*) - S_{n_s}^{(s)}(\theta_j^*,\btheta_{-j}^*,\hat{\bbeta}^{(s)},\bgamma^*) \notag \\
    &+ S_{n_s}^{(s)}(\theta_j^*,\btheta_{-j}^*,\hat{\bbeta}^{(s)},\bgamma^*) - S_{n_s}^{(s)}(\theta_j^*,\btheta_{-j}^*,\bbeta^*,\bgamma^*) \notag \\
    =& -\frac{1}{n_s}\sum_{i=1}^{n_s}\big\{Z_{si}(\hat{\bbeta}^{(s)})-\tau{X}_{si,j}\theta_j^*-\tau\bX_{si,-j}^\top\hat{\btheta}_{-j}^{(s)}\big\}\bX_{si,-j}^\top(\hat{\bgamma}^{(s)}-\bgamma^*) \notag \\
    & -\frac{\tau}{n_s}\sum_{i=1}^{n_s}\omega_{si}\bX_{si,-j}^{\top} \big(\hat{\btheta}_{-j}^{(s)}-\btheta_{-j}^*\big) + \frac{1}{n_s}\sum_{i=1}^{n_s}\big\{Z_{si}(\hat{\bbeta}^{(s)})-Z_{si}(\bbeta^*)\big\}\omega_{si} \notag \\
    =& -\frac{1}{n_s}\sum_{i=1}^{n_s}e_{si}(\hat{\bbeta}^{(s)})\bX_{si,-j}^\top(\hat{\bgamma}^{(s)}-\bgamma^*) + \frac{\tau}{n_s}\sum_{i=1}^{n_s}(\hat{\btheta}_{-j}^{(s)}-\btheta_{-j}^*)^\top\bX_{si,-j}\bX_{si,-j}^\top(\hat{\bgamma}^{(s)}-\bgamma^*) \notag \\
    & -\frac{\tau}{n_s}\sum_{i=1}^{n_s}\omega_{si}\bX_{si,-j}^{\top} \big(\hat{\btheta}_{-j}^{(s)}-\btheta_{-j}^*\big) + \frac{1}{n_s}\sum_{i=1}^{n_s}\big\{e_{si}(\hat{\bbeta}^{(s)})-e_{si}(\bbeta^*)\big\}\omega_{si}  \notag \\
    =& -\frac{1}{n_s}\sum_{i=1}^{n_s}(1-\mbE)\Big[\{e_{si}(\hat{\bbeta}^{(s)})-e_{si}(\bbeta^*)\}\bX_{si,-j}^{\top}\Big](\hat{\bgamma}^{(s)}-\bgamma^*) + \frac{1}{n_s}\sum_{i=1}^{n_s}e_{si}(\bbeta^*)\bX_{si,-j}^{\top}(\hat{\bgamma}^{(s)}-\bgamma^*) \notag \\
    &+ \mbE\{e_{si}(\hat{\bbeta}^{(s)})\bSigma_{-j}^{-1/2}\bX_{si,-j}\}^{\top} \bSigma_{-j}^{1/2}(\hat{\bgamma}^{(s)}-\bgamma^*) + \tau\big(\hat{\btheta}_{-j}^{(s)}-\btheta_{-j}^*\big)^{\top}(\hat{\bSigma}_{-j}^{(s)}-\bSigma_{-j})(\hat{\bgamma}^{(s)}-\bgamma^*) \notag \\
    &+ \tau\big(\hat{\btheta}_{-j}^{(s)}-\btheta_{-j}^*\big)^{\top}\bSigma_{-j}(\hat{\bgamma}^{(s)}-\bgamma^*) -\frac{\tau}{n_s}\sum_{i=1}^{n_s}\omega_{si}\bX_{si,-j}^{\top} \big(\hat{\btheta}_{-j}^{(s)}-\btheta_{-j}^*\big) \notag \\
    &+ \frac{1}{n_s}\sum_{i=1}^{n_s}(1-\mbE)\Big[\{e_{si}(\hat{\bbeta}^{(s)})-e_{si}(\bbeta^*)\}\omega_{si}\Big] + \mbE\{e_{si}(\hat{\bbeta}^{(s)})\omega_{si}\}.
\end{align}
By the fact $\mbE(e_{si}\bX_{si})=\mbE\{\bX_{si}\mbE(e_{si}\mid\bX_{si})\}=0$ and similarly $\mbE(e_{si}\omega_{si})=0$, under the restriction $\hat{\bbeta}^{(s)}\in\bbeta^*+\mbB_{\bSigma}(r_0)\cap\mbB_1(r_1)$, it follows that
\begin{align}
    \label{scoreapproxerror}
    & \big|S_{n_s}^{(s)}(\theta_j^*,\hat{\btheta}_{-j}^{(s)},\hat{\bbeta}^{(s)},\hat{\bgamma}^{(s)}) - S_{n_s}^{(s)}(\theta_j^*,\btheta_{-j}^*,\bbeta^*,\bgamma^*)\big| \notag \\ 
    \leq& \sup_{\bbeta\in\bbeta^*+\mbB_{1}(r_1)}\bigg\Vert\frac{1}{n_s}\sum_{i=1}^{n_s}(1-\mbE)\Big[\{e_{si}(\bbeta)-e_{si}(\bbeta^*)\}\bX_{si}\Big]\bigg\Vert_{\infty}\Vert\hat{\bgamma}^{(s)}-\bgamma^*\Vert_1 \notag \\ 
    &+ \bigg\Vert\frac{1}{n_s}\sum_{i=1}^{n_s}e_{si}(\bbeta^*)\bX_{si}\bigg\Vert_{\infty}\Vert\hat{\bgamma}^{(s)}-\bgamma^*\Vert_1 + \sup_{\bbeta\in\bbeta^*+\mbB_{\bSigma}(r_0)}\Vert\mbE\{e_{si}(\bbeta)\bSigma^{-1/2}\bX_{si}\}\Vert_2\Vert\hat{\bgamma}^{(s)}-\bgamma^*\Vert_{\bSigma_{-j}} \notag \\
    &+ \tau\Vert\hat{\btheta}_{-j}^{(s)}-\btheta_{-j}^*\Vert_1\Vert\hat{\bSigma}_{-j}^{(s)}-\bSigma_{-j}\Vert_{\max}\Vert\hat{\bgamma}^{(s)}-\bgamma^*\Vert_1 + \tau\Vert\hat{\btheta}_{-j}^{(s)}-\btheta_{-j}^*\Vert_{\bSigma_{-j}}\Vert\hat{\bgamma}^{(s)}-\bgamma^*\Vert_{\bSigma_{-j}} \notag \\
    &+ \bigg\Vert\frac{1}{n_s}\sum_{i=1}^{n_s}\omega_{si}\bX_{si} \bigg\Vert_{\infty}\tau\Vert\hat{\btheta}_{-j}^{(s)}-\btheta_{-j}^*\Vert_1 + \sup_{\bbeta\in\bbeta^*+\mbB_1(r_1)}\bigg\vert\frac{1}{n_s}\sum_{i=1}^{n_s}(1-\mbE)\Big[\{e_{si}(\bbeta)-e_{si}(\bbeta^*)\}\omega_{si}\Big]\bigg\vert \notag \\
    &+ \sup_{\bbeta\in\bbeta^*+\mbB_{\bSigma}(r_0)}\big\vert\mbE\{e_{si}(\bbeta)\omega_{si}\}\big\vert.
\end{align}
We then analyze each term in \eqref{scoreapproxerror}. Firstly, by Lemma A.2 in \cite{zhang2023hdesreg}, we know that
\begin{align*}
    \sup_{\bbeta\in\bbeta^*+\mbB_{\bSigma}(r_0)}\Vert\mbE\{e_{si}(\bbeta)\bSigma^{-1/2}\bX_{si}\}\Vert_2 \leq \frac{1}{2}F_uM_3r_0^2, \\ \sup_{\bbeta\in\bbeta^*+\mbB_{\bSigma}(r_0)}\big\vert\mbE\{e_{si}(\bbeta)\omega_{si}\}\big\vert \leq \frac{1}{2}F_uB_{\bX}r_0^2. 
\end{align*}
Taking $t=\log(p)$ in Lemma A.1 in \cite{zhang2023hdesreg}, we have with probability at least $1-2/p$ 
\begin{align*}
    \sup_{\bbeta\in\bbeta^*+\mbB_{1}(r_1)}\bigg\Vert\frac{1}{n_s}\sum_{i=1}^{n_s}(1-\mbE)\Big[\{e_{si}(\bbeta)-e_{si}(\bbeta^*)\}\bX_{si}\Big]\bigg\Vert_{\infty} &\leq C_1B_{\bX}\sigma_{\bX}\{\log(p)/n_s\}^{1/2}r_1, \\
    \sup_{\bbeta\in\bbeta^*+\mbB_1(r_1)}\bigg\vert\frac{1}{n_s}\sum_{i=1}^{n_s}(1-\mbE)\Big[\{e_{si}(\bbeta)-e_{si}(\bbeta^*)\}\omega_{si}\Big]\bigg\vert &\leq C_1B_{\bX}\sigma_{\bX}\{\log(p)/n_s\}^{1/2}r_1,
\end{align*}
provided that $n_s\geq(q_{\epsilon}/\sigma_{\epsilon})^2(B_{\bX}/\sigma_{\bX})^2\log(p)$. Similarly, by Lemma A.3 in \cite{zhang2023hdesreg}, we have with probability at least $1-2/p$
\begin{align*}
    \bigg\Vert\frac{1}{n_s}\sum_{i=1}^{n_s}e_{si}(\bbeta^*)\bX_{si}\bigg\Vert_{\infty} &\leq C_2\sigma_{\epsilon}\sigma_{\bX}\{\log(p)/n_s\}^{1/2}, \\
    \bigg\Vert\frac{1}{n_s}\sum_{i=1}^{n_s}\omega_{si}\bX_{si} \bigg\Vert_{\infty} &\leq C_2B_{\bX}^2\{\log(p)/n_s\}^{1/2}.
\end{align*}

Further note that $\Vert\hat{\bSigma}_{-j}^{(s)}-\bSigma_{-j}\Vert_{\max} \leq \Vert\hat{\bSigma}_{s}-\bSigma\Vert_{\max}$, it then follows from Lemma \ref{lemmacovest} that $\Vert\hat{\bSigma}_{s}-\bSigma\Vert_{\max}\leq 
M_4^{1/2}\sigma_{\bX}^2\sqrt{6\log(p)/n_s} + B_{\bX}^2\log(p)/n_s$ holds with probability no lower than $1-2/p$. Combining all the pieces yields that given $\hat{\bbeta}^{(s)}\in\bbeta^*+\mbB_{\bSigma}(r_0)\cap\mbB_1(r_1)$, with probability at least $1-8/p$ 
\begin{align}
    \label{scoreapproxmed}
    &\big|S_{n_s}^{(s)}(\theta_j^*,\hat{\btheta}_{-j}^{(s)},\hat{\bbeta}^{(s)},\hat{\bgamma}^{(s)}) - S_{n_s}^{(s)}(\theta_j^*,\btheta_{-j}^*,\bbeta^*,\bgamma^*)\big| \notag \\
    \leq& C_1B_{\bX}\sigma_{\bX}\Big\{\frac{\log(p)}{n_s}\Big\}^{1/2}r_1 + \frac{1}{2}F_uB_{\bX}r_0^2 + (C_1B_{\bX}r_1+C_2\sigma_{\epsilon})\sigma_{\bX}\Big\{\frac{\log(p)}{n_s}\Big\}^{1/2}\Vert\hat{\bgamma}^{(s)}-\bgamma^*\Vert_1 \notag \\
    &+ \left(\frac{1}{2}F_uM_3r_0^2+\tau\Vert\hat{\btheta}_{-j}^{(s)}-\btheta_{-j}^*\Vert_{\bSigma}\right)\Vert\hat{\bgamma}^{(s)}-\bgamma^*\Vert_{\bSigma_{-j}} \notag \\
    &+ \left[ M_4^{1/2}\sigma_{\bX}^2\Big\{\frac{\log(p)}{n_s}\Big\}^{1/2} + B_{\bX}^2\frac{\log(p)}{n_s}\right]\tau\Vert\hat{\btheta}_{-j}^{(s)}-\btheta_{-j}^*\Vert_1\Vert\hat{\bgamma}^{(s)}-\bgamma^*\Vert_1 \notag \\
    &+ C_2B_{\bX}^2\Big\{\frac{\log(p)}{n_s}\Big\}^{1/2}\tau\Vert\hat{\btheta}_{-j}^{(s)}-\btheta_{-j}^*\Vert_1.
\end{align}
Next we focus on the error bounds of estimators $\hat{\bbeta}^{(s)}$, $\hat{\btheta}^{(s)}$ and $\hat{\bgamma}^{(s)}$. By Theorem \ref{thmesterrbound}, Proposition \ref{propbetaesterr} and \ref{propgammaesterr}, given that the penalization parameters $\lambda_q^{(s)}\asymp\lambda_e^{(s)}\asymp\lambda_m^{(s)}\asymp\sqrt{\log(p)/N_s}$, the following error bounds 
\begin{align*}
    &\Vert\hat{\bbeta}^{(s)}-\bbeta^*\Vert_{\bSigma} \lesssim \{s_0\log(p)/N_s\}^{1/2}, \quad \Vert\hat{\bbeta}^{(s)}-\bbeta^*\Vert_1 \lesssim s_0\{\log(p)/N_s\}^{1/2}; \\
    &\tau\Vert\hat{\btheta}^{(s)}-\btheta^*\Vert_{\bSigma} \lesssim \{s_0\log(p)/N_s\}^{1/2}, \quad \tau\Vert\hat{\btheta}^{(s)}-\btheta^*\Vert_1 \lesssim s_0\{\log(p)/N_s\}^{1/2}; \\
    &\Vert\hat{\bgamma}^{(s)}-\bgamma^*\Vert_{\bSigma_{-j}} \lesssim \{s_r\log(p)/N_s\}^{1/2}, \quad \Vert\hat{\bgamma}^{(s)}-\bgamma^*\Vert_1 \lesssim s_r\{\log(p)/N_s\}^{1/2} 
\end{align*}
hold with probability no lower than $1-9/p$, provided that $n_1 \gtrsim \max(s_0^2,s_r)\log(p)$. Substituting these bounds into \eqref{scoreapproxmed}, we obtain that 
\begin{align*}
    \big|S_{n_s}^{(s)}(\theta_j^*,\hat{\btheta}_{-j}^{(s)},\hat{\bbeta}^{(s)},\hat{\bgamma}^{(s)}) - S_{n_s}^{(s)}(\theta_j^*,\btheta_{-j}^*,\bbeta^*,\bgamma^*)\big| \lesssim \max(s_0,s_r)\log(p)/(n_sN_s)^{1/2}
\end{align*}
holds with probability at least $1-15/p$. It then follows that
\begin{align}
    \label{E1bound}
    E_1 \leq& \frac{1}{N_b}\sum_{s=1}^{b}n_s|\hat{S}_{n_s}^{(s)}(\theta_j^*)-S_{n_s}^{*}| \lesssim \frac{\max(s_0,s_r)\log(p)}{N_b}\sum_{s=1}^{b}\sqrt{\frac{n_s}{N_s}} \notag \\ 
    \leq& \frac{\max(s_0,s_r)\log(p)\sqrt{b\{1+\log(N_b/n_1)\}}}{N_b}.   
\end{align}
Therefore, we conclude from \eqref{E1bound} that as long as  $\max(s_0,s_r)\log(p)\sqrt{b\{1+\log(N_b/n_1)\}}=o(N_b^{1/2})$, we have $E_1=o_p(N_b^{-1/2})$, and it follows that 
\begin{align*}
   \tau N_b^{1/2}\big(\tilde{\theta}_j^{(b)}-\theta_j^{*}\big)  \stackrel{p} {\rightarrow} N(0,\sigma_{I}^2/\sigma_{\omega}^4).
\end{align*}
\end{proof}

\begin{proof}[Proof of Theorem \ref{theoremconsvarest}.]
    
To begin with, recall from Proposition \ref{propgammaesterr} that $\Vert\hat{\bgamma}^{(s)}-\bgamma^*\Vert_1=o_p(1)$ under condition $s_r\log(p)=o(n_1)$. By the fact that $\Vert\bgamma^*\Vert_1$ is bounded and $\Vert\bX\Vert_{\infty}\leq B_{\bX}$, using Lemma \ref{lemmaquadraformbound} we have
\begin{align*}
    \big|\hat{\omega}_{si}^2-\omega_{si}^2| =& |(\hat{\bgamma}^{(s)}-\bgamma^*)^\top\bX_{si,-j}\bX_{si,-j}^\top(\hat{\bgamma}^{(s)}-\bgamma^*) + 2\bgamma^*\bX_{si,-j}\bX_{si,-j}^\top(\hat{\bgamma}^{(s)}-\bgamma^*) \big| \\
    \leq& \Vert\bX_{si,-j}\bX_{si,-j}^\top\Vert_{\max}\Vert\hat{\bgamma}^{(s)}-\bgamma^*\Vert_1^2 + 2\Vert\bX_{si,-j}\bX_{si,-j}^\top\Vert_{\max}\Vert\bgamma^*\Vert_1\Vert\hat{\bgamma}^{(s)}-\bgamma^*\Vert_1 \\
    =& o_p(1).
\end{align*}
It then follows that 
\[|\hat{\sigma}_{\omega,(b)}^2-\hat{\sigma}_{\omega}^2|\leq (N_b-\hat{s}_m^{(b)})^{-1}\sum_{s=1}^{b}\sum_{i=1}^{n_s}|\hat{\omega}_{si}^2-\omega_{si}^2| \stackrel{p}{\rightarrow} 0 \  \ \text{as} \ \ N_b\to\infty. \] 
Turn to the consistency of $\hat{\sigma}_{I,(b)}^2$. Denote $\hat{M}_b=N_b-\hat{s}_e^{(b)}-\hat{s}_q^{(b)}-\hat{s}_m^{(b)}$, we consider the decomposition
\begin{align}
    \label{sigIesterr}
    \big|\hat{\sigma}_{I,(b)}^2-\sigma_{I}^2\big| =& \bigg|\frac{1}{M_b}\sum_{s=1}^{b}\sum_{i=1}^{n_s}\hat{\omega}_{si}^2 e_{si}(\hat{\bbeta}^{(s)})^2-\mbE\{\omega_{si}^2 e_{si}(\bbeta^*)^2\}\bigg| \notag \\
    \leq& \bigg|\frac{1}{M_b}\sum_{s=1}^{b}\sum_{i=1}^{n_s}(\hat{\omega}_{si}^2-\omega_{si}^2)e_{si}(\hat{\bbeta}^{(s)})^2\bigg| + \bigg|\frac{1}{M_b}\sum_{s=1}^{b}\sum_{i=1}^{n_s}\omega_{si}^2\{e_{si}(\hat{\bbeta}^{(s)})^2-e_{si}(\bbeta^*)^2\}\bigg| \notag \\
    &+ \bigg|\frac{1}{M_b}\sum_{s=1}^{b}\sum_{i=1}^{n_s}(1-\mbE)e_{si}(\bbeta^*)^2\omega_{si}^2\bigg|.
\end{align}
We now analyze the three terms in \eqref{sigIesterr}. For the first term, by the same argument in the proof of Theorem 4.4 in \cite{zhang2023hdesreg}, we have $\max_{i=1,\dots,n_s}|\hat{\omega}_{si}^2-\omega_{si}^2|=o_p(1)$ and ${n_s}^{-1}\sum_{i=1}^{n_s}e_{si}(\hat{\bbeta}^{(s)})^2=O_p(1)$, for $s=1,\dots,b$ under the condition $\max(s_0^2,s_r)\log(p)=o(n_1)$. Combining all these results from $s=1$ to $s=b$, we know that  $\max_{s=1,\dots,b}\max_{i=1,\dots,n_s}|\hat{\omega}_{si}^2-\omega_{si}^2|=o_p(1)$ and $M_b^{-1}\sum_{s=1}^{b}\sum_{i=1}^{n_s}e_{si}(\hat{\bbeta}^{(s)})^2=O_p(1)$. It then follows that the first term $|M_b^{-1}\sum_{s=1}^{b}\sum_{i=1}^{n_s}(\hat{\omega}_{si}^2-\omega_{si}^2)e_{si}(\hat{\bbeta}^{(s)})^2|=o_p(1)$. 

Similarly, from the proof of Theorem 4.4 in \cite{zhang2023hdesreg} we know that $|n_s^{-1}\sum_{i=1}^{n_s}\{e_{si}(\hat{\bbeta}^{(s)})^2-e_{si}(\bbeta^{*})^2\}|=o_p(1)$ for $s=1,\dots,b$ under the condition $s_r\log(p)=o(n_1)$. Then we have $|M_b^{-1}\sum_{s=1}^{b}\sum_{i=1}^{n_s}\{e_{si}(\hat{\bbeta}^{(s)})^2-e_{si}(\bbeta^{*})^2\}|=o_p(1)$. Moreover, $\max_{s=1,\dots,b}\max_{i=1,\dots,n_s}|\omega_{si}|$ is bounded by Assumption \ref{assumpomega}. Thus the second term is $o_p(1)$. The third term is $o_p(1)$ by the weak law of large number. Combining all the above  results with \eqref{sigIesterr} yields that $|\hat{\sigma}_{I,(b)}^2-\sigma_{I}^2|=o_p(1)$.

Further note that the proposed estimator $\hat{\sigma}_{\omega,(b)}^2$ shares the same form with the lasso variance estimator $\sigma_{L}^2$ in \cite{fan2012varest}. Thus we demonstrate the asymptotic normality of $\hat{\sigma}_{\omega,(b)}^2$ by applying Theorem 3 of \cite{fan2012varest}. It is necessary to confirm the conditions they imposed. Firstly, by making a slight modification to their proof, Assumption 1 in \cite{fan2012varest} could be reduced to the condition $\mbE(\omega_{si}\bX_{si,-j})=0$, which is stated in our model \eqref{modelprojreg}. Next, according to the proof of Theorem 4.4 in \cite{zhang2023hdesreg},  Assumption 2 in \cite{fan2012varest} holds under our Assumption \ref{assumpxboundRME}. Moreover, Assumption 3 in \cite{fan2012varest} is the same as $\Vert\bX\Vert_{\infty}\leq B_{\bX}$ in our Assumption \ref{assumpxboundRME} and Assumption 4 in \cite{fan2012varest} holds due to the sub-exponential setting of $\varepsilon$ in our Assumption \ref{assumpomega}. It remains to verify Assumption 7 in \cite{fan2012varest}.  
Denote the design matrix up to the b-th batch as $\tilde{\mbX}_{b}$. Let $M\subset\{1,\dots,p\}$ be an index set with cardinality $|M|=m$ and let $\tilde{\mbX}_{b,M} \in \mbR^{N_b\times m}$ be the submatrix with columns in $M$. We aim to show that
there exist constants $0<\mu_{\min}\leq\mu_{\max}<\infty$ such that
\begin{align}
    \label{assump7eigenineq}
    &P\Big\{\min_{m\leq s_m\log(N_b)}\Lambda_{\min}(N_b^{-1}\tilde{\mbX}_{b,M}^{\top}\tilde{\mbX}_{b,M}) \geq \mu_{\min} \Big\} \to 1 \quad \text{and} \notag \\
    &P\Big\{\max_{m\leq s_m+\min\{N_b,p\}}\Lambda_{\max}(N_b^{-1}\tilde{\mbX}_{b,M}^{\top}\tilde{\mbX}_{b,M}) \leq \mu_{\max} \Big\} \to 1
\end{align}
as $N_b\to\infty$. By Assumption \ref{assumpxboundRME} and \ref{assumpepsconddense}, we derive that 
\begin{align*}
    \Lambda_{\min}(N_b^{-1}\tilde{\mbX}_{b,M}^{\top}\tilde{\mbX}_{b,M}) =& \min_{\Vert\bu\Vert_2=1}\bu^\top(N_b^{-1}\tilde{\mbX}_{b,M}^{\top}\tilde{\mbX}_{b,M})\bu \\
    =& \min_{\Vert\bu\Vert_2=1}\big[\bu^\top\{N_b^{-1}(1-\mbE)\tilde{\mbX}_{b,M}^{\top}\tilde{\mbX}_{b,M}\}\bu + \bu^\top\mbE(N_b^{-1}\tilde{\mbX}_{b,M}^{\top}\tilde{\mbX}_{b,M})\bu \big] \\
    \geq& -\Vert N_b^{-1}(1-\mbE)\tilde{\mbX}_{b,M}^{\top}\tilde{\mbX}_{b,M}\Vert + \Lambda_{\min}\big\{\mbE(N_b^{-1}\tilde{\mbX}_{b,M}^{\top}\tilde{\mbX}_{b,M})\big\} \\
    \geq& -\Vert N_b^{-1}(1-\mbE)\tilde{\mbX}_{b,M}^{\top}\tilde{\mbX}_{b,M}\Vert + \Lambda_p.
\end{align*}
Similarly, we have $\Lambda_{\max}(N_b^{-1}\tilde{\mbX}_{b,M}^{\top}\tilde{\mbX}_{b,M})\leq\Vert N_b^{-1}(1-\mbE)\tilde{\mbX}_{b,M}^{\top}\tilde{\mbX}_{b,M}\Vert + \Lambda_1$. For the term $\Vert N_b^{-1}(1-\mbE)\tilde{\mbX}_{b,M}^{\top}\tilde{\mbX}_{b,M}\Vert$,  
\cite{zhang2023hdesreg} has shown in the proof that for any $t>0$,
\begin{align*}
    P\bigg\{\Vert N_b^{-1}(1-\mbE)\tilde{\mbX}_{b,M}^{\top}\tilde{\mbX}_{b,M}\Vert>2\Lambda_1\sqrt{\frac{2\kappa_4 t}{N_b}}+2B_{\bX}^2\frac{mt}{N_b}\bigg\} < 9^me^{-t}.
\end{align*}
The above bound holds for any fixed $m$. Further denote $w_{N_b}=s_m\log(N_b)$ and $v_{N_b}=s_m+\min\{p,N_b\}$. Taking the union bound over all set $M$ satisfying $|M|\leq w_{N_b}$ yields
\begin{align*}
    &P\bigg\{\min_{m\leq w_{N_b}}\Lambda_{\min}(N_b^{-1}\tilde{\mbX}_{b,M}^{\top}\tilde{\mbX}_{b,M})\leq -2\Lambda_1\sqrt{\frac{2\kappa_4 t}{N_b}}-2B_{\bX}^2\frac{w_{N_b}t}{N_b}+\Lambda_p\bigg\} \\
    <& \sum_{m=0}^{\lfloor w_{N_b}\rfloor}\binom{p}{m} 9^me^{-t} \leq \Big(\frac{9ep}{w_{N_b}}\Big)^{w_{N_b}}e^{-t}.
\end{align*}
Then choosing $t=w_{N_b}\log(9epw_{N_b}^{-1})+\log(N_b)$ and under the condition $s_m^2\log^2(N_b)\log(p)=o(N_b)$ we have
\begin{align*}
    \min_{m\leq w_{N_b}}\Lambda_{\min}(N_b^{-1}\tilde{\mbX}_{b,M}^{\top}\tilde{\mbX}_{b,M}) \geq& -2\sqrt{2\kappa_4}\Lambda_1 \bigg\{\frac{\log(N_b)+w_{N_b}\log(9epw_{N_b}^{-1})}{N_b}\bigg\}^{1/2} \\
    & -2w_{N_b}B_{\bX}^2\frac{\log(N_b)+w_{N_b}\log(9epw_{N_b}^{-1})}{N_b} + \Lambda_p \\
    \geq& \Lambda_p/2,
\end{align*}
with probability at least $1-N_b^{-1}$. Employing similar arguments, we could derive with probability at least $1-N_b^{-1}$,
\begin{align*}
    \min_{m\leq v_{N_b}}\Lambda_{\min}(N_b^{-1}\tilde{\mbX}_{b,M}^{\top}\tilde{\mbX}_{b,M}) \leq& 2\sqrt{2\kappa_4}\Lambda_1 \bigg\{\frac{\log(N_b)+v_{N_b}\log(9epv_{N_b}^{-1})}{N_b}\bigg\}^{1/2} \\
    & +2v_{N_b}B_{\bX}^2\frac{\log(N_b)+v_{N_b}\log(9epv_{N_b}^{-1})}{N_b} + \Lambda_1 \\
    \leq& 2\Lambda_1,
\end{align*}
provided that $s_m^2p^2=o(N_b)$. Therefore, \eqref{assump7eigenineq} holds with $\mu_{\min}=\Lambda_p/2$ and $\mu_{\max}=2\Lambda_1$.
Finally, applying Theorem 3 in \cite{fan2012varest}, we have
$\sqrt{N_b}\big(\hat{\sigma}_{\omega,(b)}^2 - \sigma_{\omega}^2\big) \stackrel{d}{\rightarrow} N(0,\mbE\omega_{si}^4-\sigma_{\omega}^4)$.

\end{proof}

\section{Technical Lemmas}
\begin{lemma}\label{lemmacovest}
    Under Assumption \ref{assumpxboundRME}, we have
    \begin{align*}
        P\left[ \Vert\hat{\bSigma}_s-\bSigma\Vert_{\max} \geq M_4^{1/2}\sigma_{\bX}^2\Big\{\frac{6\log(p)}{n_s}\Big\}^{1/2} + B_{\bX}^2\frac{\log(p)}{n_s} \right] &\leq \frac{2}{p}, \\
        P\left[ \Vert\tilde{\bSigma}_b-\bSigma\Vert_{\max} \geq M_4^{1/2}\sigma_{\bX}^2\Big\{\frac{6\log(p)}{N_b}\Big\}^{1/2} + B_{\bX}^2\frac{\log(p)}{N_b} \right] &\leq \frac{2}{p}.
    \end{align*}
    \begin{proof}
        Note that $|X_{s,ij}X_{s,ik}|\leq B_{\bX}^2$ almost surly. Moreover, for every pair $1\leq j\leq k\leq p$, by Cauchy-schwartz inequality, 
        \begin{align*}
            \mbE\big\{(X_{s,ij}X_{s,ik})^2\big\} \leq \mbE(X_{s,ij}^4)^{1/2}\mbE(X_{s,ik}^4)^{1/2} \leq M_4^{1/2}\mbE(X_{s,ij}^2)M_4^{1/2}\mbE(X_{s,ik}^2) \leq M_4\sigma_{\bX}^4.
        \end{align*}
        Applying Bernstein's inequality (Theorem 2.10 in \cite{Boucheron2013concineqnonasym}), taking $v=n_sM_4\sigma_{\bX}^4, \ c=B_{\bX}^2/3$ and $t=3\log(p)$, it then follows that 
        \begin{align*}
            P\left[\Big|\frac{1}{n_s}\sum_{i=1}^{n_s}(X_{s,ij}X_{s,ik}-\sigma_{ij})\Big| \geq  M_4^{1/2}\sigma_{\bX}^2\Big\{\frac{6\log(p)}{n_s}\Big\}^{1/2} + B_{\bX}^2\frac{\log(p)}{n_s} \right] \leq \frac{2}{p^3}, \\
            P\left[\Big|\frac{1}{N_b}\sum_{s=1}^b\sum_{i=1}^{n_s}(X_{s,ij}X_{s,ik}-\sigma_{ij})\Big| \geq  M_4^{1/2}\sigma_{\bX}^2\Big\{\frac{6\log(p)}{N_b}\Big\}^{1/2} + B_{\bX}^2\frac{\log(p)}{N_b} \right] \leq \frac{2}{p^3}.
        \end{align*}
        Then, by the union bound over $1\leq j\leq k\leq p$, the two inequalities are demonstrated.
    \end{proof}
\end{lemma}

\begin{lemma}
	\label{lemmaquadraformbound}
	For two vectors $\bu=(u_1,\dots, u_p)^\top\in\mbR^p$, $\bv=(v_1, \dots, v_p)^\top\in\mbR^p$ and a matrix $\mbA=(a_{ij})_{i,j}^p\in \mbR^{p\times p}$, we have 
	\begin{align*}
		&\Vert\mbA\bu\Vert_{\infty} \leq \Vert\mbA\Vert_{\max}\Vert\bu\Vert_1, \\
		&\bu^\top\mbA\bv\leq \Vert\mbA\Vert_{\max}\Vert\bu\Vert_1\Vert\bv\Vert_1. 
	\end{align*}    
\end{lemma}
{\bf Proof.} Firstly, we have
\begin{align*}
	\Vert\mbA\bu\Vert_{\infty} = \max\limits_{1\leq i\leq p} \left|\sum_{j=1}^p a_{ij}u_j\right| 
	\leq \max\limits_{1\leq i\leq p} \max\limits_{1\leq j\leq p} |a_{ij}|\sum_{j=1}^p |u_j|  = \Vert\mbA\Vert_{\max}\Vert\bu\Vert_1.
\end{align*}
Then, by H{\"o}lder's inequality
\begin{align*}
	\bu^\top\mbA\bv \leq \Vert\mbA\bu\Vert_{\infty} \Vert\bv\Vert_1 \leq \Vert\mbA\Vert_{\max}\Vert\bu\Vert_1\Vert\bv\Vert_1.
\end{align*}

\begin{lemma}
    \label{lemmabatchsize}
    Let $n_s$ and $N_s$ be the batch size and the cumulative batch size, respectively when the $s$-th data batch arrives, $s=1,2,...,b$. Then we have
    \[
    \sum_{s=1}^b \frac{n_s}{N_s}\leq 1+\log\left(\frac{N_b}{n_1}\right), \qquad  \sum_{s=1}^b \frac{n_s}{\sqrt{N_s}} \leq 2\sqrt{N_b}, \qquad 
    \sum_{s=1}^b \sqrt{\frac{n_s}{N_s}}\leq \sqrt{b+b\log\left(\frac{N_b}{n_1}\right)}.
    \]
\end{lemma}
This lemma is widely used in online learning and has been confirmed in the literature, see for example \citep{han2021delasso, luo2021glm}.

\section{Additional numerical results}
\label{app:numres}

\begin{figure}
    \centering
\includegraphics[width=0.9\linewidth]{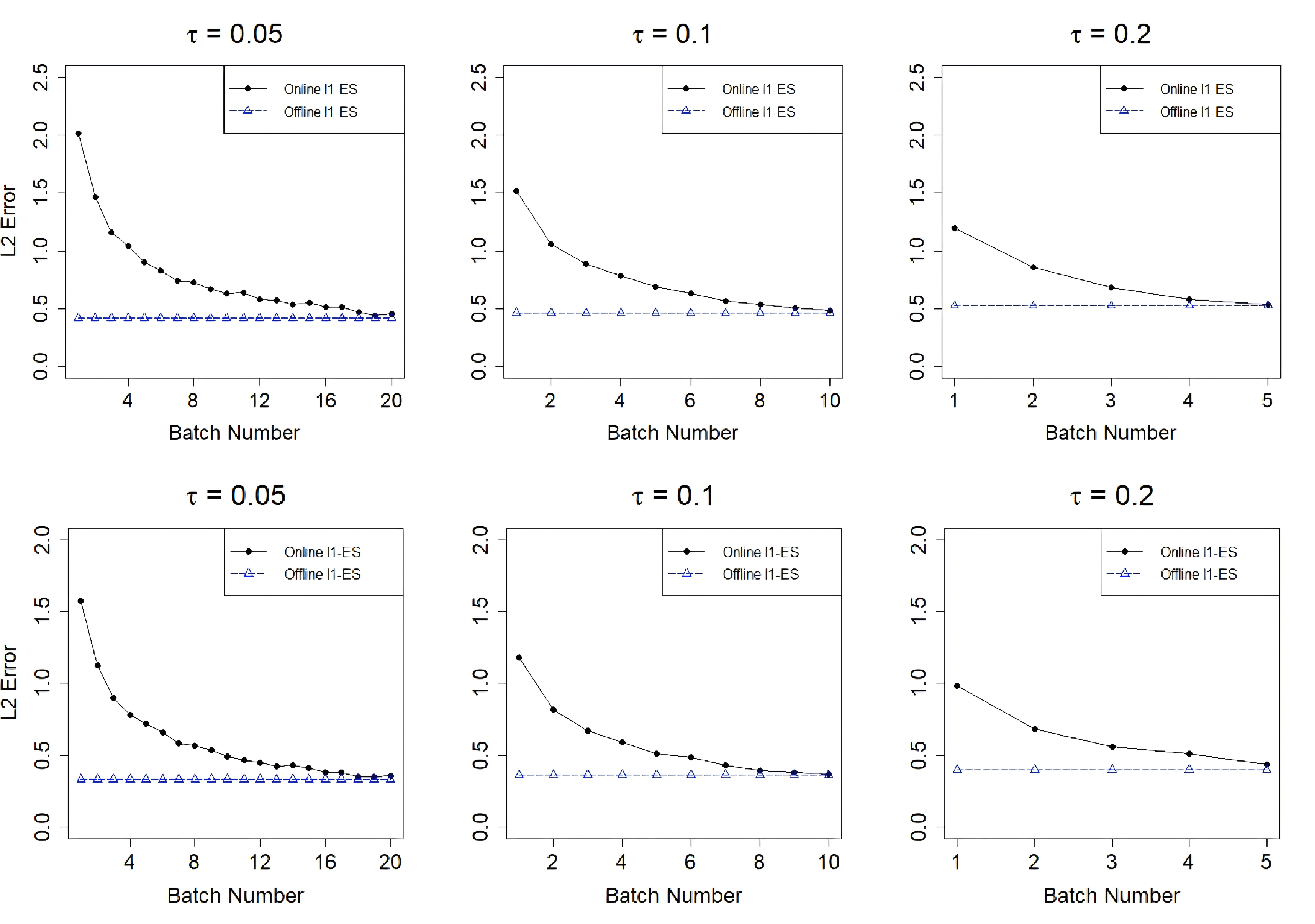}
\caption{Estimation $\ell_2$-errors on each batch under Covariate 2}
    \label{fig2}
\end{figure}

\begin{figure}
    \centering   \includegraphics[width=0.9\linewidth]{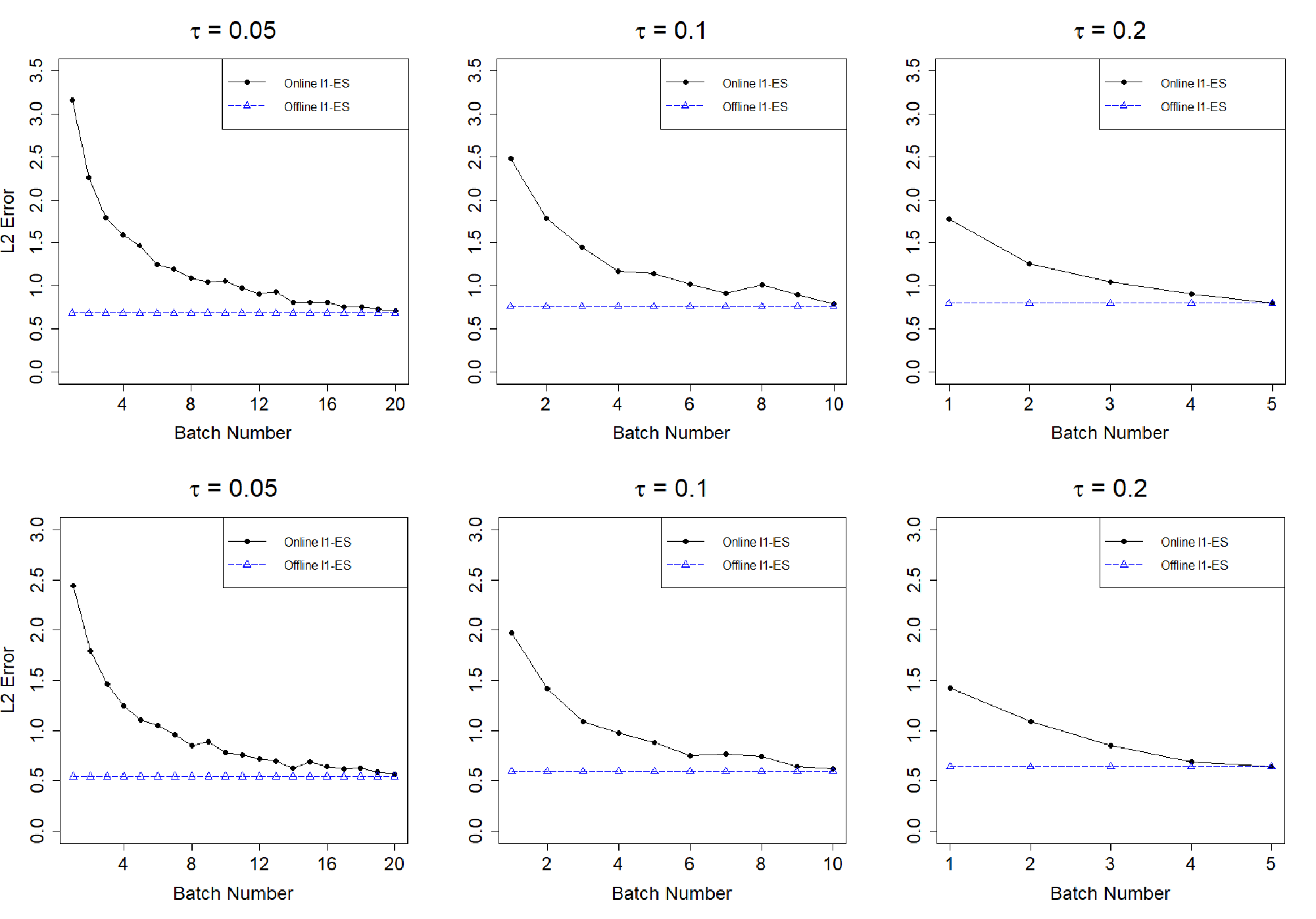}    \caption{Estimation $\ell_2$-errors on each batch under Covariate 3}
    \label{fig3}
\end{figure}

\end{appendices}

\newpage
\normalem
\bibliographystyle{apalike}
\bibliography{bib_OES}

\end{document}